\documentclass[a4paper,UKenglish,cleveref,autoref,numberwithinsect]{lipics-v2019}

\title{Domain-Aware Session Types  (Extended Version)} 

\hideLIPIcs
\nolinenumbers

\titlerunning{Domain-Aware Session Types}


\author{Lu\'{i}s Caires}
           {Universidade Nova de Lisboa}
           {}{0000-0002-3215-6734}{}
\author{Jorge A. P\'{e}rez}
   {University of Groningen}
   {}{https://orcid.org/0000-0002-1452-6180}{}
   \author{Frank Pfenning}
           {Carnegie Mellon University}
           {}{}{}
\author{Bernardo Toninho}
           {Universidade Nova de  
           Lisboa}
           {}{https://orcid.org/0000-0002-0746-7514}{}

\funding{Caires and Toninho are supported by NOVA LINCS (Ref. UID/CEC/04516/2019). 
P\'{e}rez is supported by the 
NWO 
VIDI Project No. 016.Vidi.189.046. 
 Pfenning is supported by NSF Grant No. CCF-1718267: ``Enriching Session Types for Practical Concurrent Programming''.
 }
 

\Copyright{Caires et al.}
\authorrunning{Caires et al.}

\usepackage{mathpartir,xspace,color}
\usepackage{todonotes}

\newcommand{\defref}[1]{Def.~\ref{#1}\xspace}
\newcommand{\secref}[1]{\S\,\ref{#1}\xspace}
\newcommand{\propref}[1]{Prop.~\ref{#1}\xspace}
\newcommand{\figref}[1]{Fig.~\ref{#1}\xspace}
\newcommand{\thmref}[1]{Thm.~\ref{#1}\xspace}

\newcommand{\lend}{\ensuremath{\mathsf{end}}\xspace}
\newcommand{\gend}{\ensuremath{\mathtt{end}}\xspace}
\newcommand{\gto}[2]{\ensuremath{\mathtt{#1}\!\twoheadrightarrow\!\mathtt{#2}{:}}}
\newcommand{\gmoves}[5]{\ensuremath{#1\,\mathsf{moves}\,#2\,\mathsf{to}\,#3\,\mathsf{for}\,#4 \,;\,#5}}
\newcommand{\lb}[1]{\ensuremath{\mathit{#1}}}
\newcommand{\pt}[1]{\ensuremath{\mathtt{#1}}}
\newcommand{\ptset}[1]{\ensuremath{\widetilde{#1}}}

\newcommand{\partp}[1]{\ensuremath{\mathsf{part}(#1)}}
\newcommand{\npart}[1]{\ensuremath{\mathsf{npart}(#1)}}
\newcommand{\proj}[2]{\ensuremath{#1\!\!\upharpoonright\!{#2}}}

\newcommand{\fuse}{\sqcup}
\newcommand{\myfuse}{\circ}
\newcommand{\ltfuse}{\myfuse}
\newcommand{\prfuse}{\myfuse}
\newcommand{\encp}[1]{\llbracket #1 \rrbracket}
\newcommand{\ether}[1]{\ensuremath{\mathsf{M}\encp{#1}}} 
\newcommand{\etherp}[3]{\ensuremath{\mathsf{M}^{#3}\encp{#1}(#2)}}
\newcommand{\mycasebig}[3]{\ensuremath{#1\triangleright\!\big\{#2\big\}_{#3}}}
\newcommand{\myoplus}[2]{\ensuremath{\oplus\!\{#1\}_{#2}}}
\newcommand{\mywith}[2]{\ensuremath{\with\!\{#1\}_{#2}}}
\newcommand{\mysel}[2]{\ensuremath{#1\triangleleft\!#2}}
\newcommand{\lt}[1]{\ensuremath{\langle\!\langle #1\rangle\!\rangle}}
\newcommand{\fuses}{\sqcup}
\newcommand{\subt}{\preceq^{\fuses}_{\here{}}}

\newcommand{\dom}[1]{\ensuremath{dom(#1)}}

\definecolor{jpblue}{RGB}{0,0,153}

\definecolor{lgray}{gray}{0.95}
\newcommand{\redd}[1]{\textcolor{jpblue}{#1}}
\newcommand{\bluee}[1]{{#1}}

\newcommand{\btnote}[1]{\textcolor{black}{#1}}
\newcommand{\jpnote}[1]{\textcolor{black}{#1}}

\newcommand{\seq}{\mathcal{S}}
\newcommand{\rel}{\mathcal{R}}
\newcommand{\relind}[4]{\ensuremath{#1 \vdash #3 \, \rel \,  #4\,{::}\,#2}}

\newcommand{\Names}{\Lambda}
\newcommand{\para}{\mathord{\;\boldsymbol{|}\;}}
\newcommand{\linkr}[2]{[#1\!\leftrightarrow\!#2]}	
\newcommand{\send}[2]{\mathord{#1\!\left<#2\right>}}
\newcommand{\sep}{\;|\;}
\newcommand{\bn}[1]{\mbox{\it bn}(#1)}

\newcommand{\lol}[0]{\lolli} 
\newcommand{\DILL}{\ensuremath{\mathtt{DILL}}\xspace}

\newcommand{\textmigrate}{move}

\newcommand{\lmigrate}{\ensuremath{y@\omega}}

\newcommand{\here}[1]{\ensuremath{\downarrow\! #1}}

\newcommand{\wtag}{\ensuremath{w}}

\newcommand{\labelset}{\ensuremath{\lambda}}

\newcommand{\cbox}[1]{\vspace{1ex}\centerline{#1}\vspace{1ex}} 
\newcommand{\til}[1]{\widetilde{#1}}

\newcommand{\aand}{\wedge}
\newcommand{\nub}{\boldsymbol{\nu}}

\newcommand{\cpy}{\mathsf{copy}}
\newcommand{\lolli}{\multimap}
\newcommand{\tensor}{\otimes}
\newcommand{\with}{\mathbin{\&}}
\newcommand{\one}{\mathbf{1}}
\newcommand{\ov}[1]{\overline{#1}}
\newcommand{\zero}{\mathbf{0}}
\newcommand{\bang}{\mathbf{!}}

\newcommand{\D}{\Delta}
\newcommand{\G}{\Gamma}
\newcommand{\Om}{\Omega}

\newcommand{\inl}{\mathsf{inl}}
\newcommand{\inr}{\mathsf{inr}}
\def\subst#1#2{\{\raisebox{.5ex}{\small$#1$}\! / \mbox{\small$#2$}\}}

\newcommand{\fn}[1]{\mbox{\it fn}(#1)}
\newcommand{\red}{\rightarrow} 
\newcommand{\wlabel}{\mathcal{L}}

\newcommand{\outp}[2]{\ov{#1}\!\left<#2\right>\!}
\newcommand{\out}[1]{\langle #1\rangle}

\newcommand{\cut}{\mathsf{cut}}

\newcommand{\lft}[1]{{{#1}\mathsf{L}}}
\newcommand{\rgt}[1]{{{#1}\mathsf{R}}}

\newcommand{\name}[1]{\mbox{(#1)}}
\newcommand{\inp}[2]{#1(#2).}

\newcommand{\m}[1]{\mathsf{#1}}

\newcommand{\cutbang}{{\mathsf{cut}^\bang}}

\newcommand{\ext}{\leadsto}

\newcommand{\tr}[3]{#1\stackrel{#2}{\rightarrow} #3}

\newcommand{\candit}{\!\Leftrightarrow\!}

\newcommand{\tra}[1]{\xrightarrow{#1}}
\newcommand{\wtra}[1]{\stackrel{#1}{\Longrightarrow}}
\newcommand{\logsim}[1]{\,{\approx_{#1}}\,}
\newcommand{\logeq}[5]{#1 \logsim{#5} #2 :: #3 \big[#4\big]}

\newcommand{\barb}[2]{#1\!\uparrow^{\,#2}}
\newcommand{\wbarb}[2]{#1\!\Uparrow^{\,#2}}
\newcommand{\Dia}{\Diamond}

{\end{array}\]}

\usepackage{graphicx}
\usepackage{amsmath}
\usepackage{stmaryrd}
\usepackage{amssymb}
\usepackage{proof}
\usepackage{hyperref}
\usepackage{enumerate}
\usepackage{thm-restate}

\EventEditors{Wan Fokkink and Rob van Glabbeek}
\EventNoEds{2}
\EventLongTitle{30th International Conference on Concurrency Theory (CONCUR 2019)}
\EventShortTitle{CONCUR 2019}
\EventAcronym{CONCUR}
\EventYear{2019}
\EventDate{August 27--30, 2019}
\EventLocation{Amsterdam, the Netherlands}
\EventLogo{}
\SeriesVolume{140}
\ArticleNo{35}

\begin{document}

\maketitle





\begin{abstract}
  
  We develop a generalization of existing Curry-Howard interpretations
  of (binary) session types by relying on an extension of linear logic
  with features from \emph{hybrid logic}, in particular modal worlds
  that indicate \emph{domains}.  These worlds govern \emph{domain migration},
  subject to a parametric accessibility relation familiar from the
  Kripke semantics of modal logic.  The result is an expressive new
  typed process framework for domain-aware, message-passing
  concurrency. Its logical foundations ensure that well-typed
  processes enjoy session fidelity, global progress, and
  termination. Typing also ensures that processes only communicate
  with accessible domains and so respect the accessibility relation.

  Remarkably, our domain-aware framework can specify scenarios in
  which domain information is available only at runtime; flexible
  accessibility relations can be cleanly defined and statically
  enforced.  As a specific application, we introduce domain-aware {\em
    multiparty session types}, in which global protocols can express
  arbitrarily nested sub-protocols via domain migration.  We develop a
  precise analysis of these multiparty protocols by reduction to our
  binary domain-aware framework: complex domain-aware protocols can be
  reasoned about at the right level of abstraction, ensuring also the
  principled transfer of key correctness properties from the binary to
  the multiparty setting.

\end{abstract}



\section{Introduction}\label{sect:intro}
The goal of this paper is to show how existing Curry-Howard
interpretations of session
types~\cite{DBLP:conf/concur/CairesP10,DBLP:journals/mscs/CairesPT16}
can be generalized to a \emph{domain-aware} setting by relying on an
extension of linear logic with features from \emph{hybrid logic}
\cite{DBLP:journals/entcs/Reed07,braunerdepaiva06}.  
These extended logical foundations of message-passing
concurrency allow us to analyze complex domain-aware concurrent
systems (including those governed by multiparty protocols) in a
precise and principled manner.

Software systems typically rely on \emph{communication} between
heterogeneous services; at their heart, these systems rely on
message-passing protocols that combine mobility, concurrency, and
distribution.  As distributed services are often virtualized,
protocols should span diverse software and hardware \emph{domains}.
These domains can have multiple interpretations, such as the location
where services reside, or the principals on whose behalf they act.
Concurrent behavior is then increasingly \emph{domain-aware}: a
partner's potential for interaction is influenced not only by the
domains it is involved in at various protocol phases (its context),
but also by \emph{connectedness} relations among domains.  Moreover,
domain architectures are rarely fully specified: to aid modularity and
platform independence, system participants (e.g., developers, platform
vendors, service clients) often have only partial views of actual
domain structures.  Despite their importance in communication
correctness and trustworthiness at large, the formal status of domains
within \btnote{\em typed} models of message-passing systems remains unexplored.

This paper contributes to typed approaches to the analysis of
domain-aware communications, with a focus on
\emph{session-based concurrency}.  
This approach specifies the intended message-passing protocols 
 as \emph{session
  types}~\cite{DBLP:conf/concur/Honda93,DBLP:conf/esop/HondaVK98,DBLP:conf/wsfm/Dezani-Ciancaglinid09}. 
Different  type theories for \emph{binary} and \emph{multiparty} ($n$-ary) protocols have been developed.
In both cases,  typed specifications 
can be conveniently coupled with 
$\pi$-calculus processes~\cite{MilnerPW92a}, in which so-called session channels connect exactly two subsystems.
Communication correctness usually results from 
two properties: \emph{session fidelity} (type preservation) and
\emph{deadlock freedom} (progress).  The former says that
well-typed processes always evolve to well-typed processes (a safety
property); the latter says that well-typed processes will never get
into a stuck state (a liveness property).  

A key motivation for this paper is the sharp contrast between (a)~the
growing relevance of domain-awareness in message-passing, concurrent
systems and (b)~the expressiveness of existing session type
frameworks, binary and multiparty, which cannot adequately specify
(let alone enforce) domain-related requirements.  Indeed, existing
session types frameworks, including those based on Curry-Howard
interpretations
\cite{DBLP:conf/concur/CairesP10,DBLP:journals/jfp/Wadler14,DBLP:conf/concur/CarboneLMSW16},
capture communication behavior at a level of abstraction in which even
basic domain-aware assertions (e.g., ``\emph{Shipper} resides in
domain \texttt{AmazonUS}'') cannot be expressed.  As an unfortunate
consequence, the effectiveness of the analysis techniques derived from
these frameworks is rather limited.

To better illustrate our point, consider a common distributed design pattern: 
a middleware agent ($\pt{mw}$) which answers requests
from clients ($\pt{cl}$), sometimes offloading the requests to a
server ($\pt{serv}$) to better manage local resource availability. In the
framework of multiparty session types \cite{DBLP:conf/popl/HondaYC08}
\jpnote{this protocol can be represented as the global type}:

\vspace{-0.2cm}
{\small\begin{tabbing}
   $\gto{cl}{mw}\{$$\lb{request}\langle req\rangle.$
   \=$\gto{mw}{cl}\{$ $\lb{reply}\langle ans
   \rangle.$ \=$\gto{mw}{serv}\{\lb{done}. \gend\} ~\mathbf{,}~$
   $\lb{wait}. \gto{mw}{serv}\{\lb{req}\langle data
   \rangle.$\\
   \>\>$\gto{serv}{mw}\{\lb{reply}\langle ans \rangle. {\gto{mw}{cl}\{\lb{reply}\langle ans \rangle.\gend\}} \}  \}\}\} $
\end{tabbing}}
\vspace{-0.2cm}

\noindent The client first sends a request to the middleware, which
answers back with either a $\lb{reply}$ message containing the answer
or a $\lb{wait}$ message, signaling that the server will be contacted
to produce the final $\lb{reply}$. 
While this multiparty protocol captures the intended communication behavior, it does not capture that protocols for the middleware and the server often involve some form of privilege escalation
or specific 
authentication---ensuring, e.g., that
the server interaction is adequately isolated from the
client, or that the escalation must precede the server interactions.
These requirements simply cannot be represented in existing
frameworks. 

Our work addresses this crucial limitation by generalizing
Curry-Howard interpretations of session types by appealing to hybrid
logic features. We develop a logically motivated typed
process framework in which \emph{worlds} from modal logics  
precisely and uniformly define the notion of \btnote{\em domain} in
session-based concurrency. At the level of {\em binary} sessions, domains
manifest themselves through point-to-point domain migration and
communication. In {\em multiparty} sessions, domain
migration is specified choreographically through the new construct
$\gmoves{\pt{p}}{\tilde{\pt{q}}}{\omega}{G_1}{G_2}$, where
participant $\pt{p}$ leads a migration of participants
$\tilde{\pt{q}}$ to domain $\omega$ in order to perform protocol $G_1$, who then
migrate back to perform protocol $G_2$.

Consider the global type
$\mathit{Offload}\triangleq
  \gto{mw}{serv}\{\lb{req}\langle data
  \rangle. \gto{serv}{mw}\{\lb{reply}\langle ans \rangle.\gend\}  \}$
   in our previous example.
Our framework allows us to refactor the global type above as:
 
  \vspace{-0.2cm}
 {\small\begin{tabbing}
   $\gto{cl}{mw}\{\lb{request}\langle req\rangle.$
   $\gto{mw}{cl}\{$ \=$\lb{reply}\langle ans
   \rangle.\gto{mw}{serv}\{\lb{done}. \gend\}\, ~\mathbf{,}~ $
   $\lb{wait}. \gto{mw}{serv}\{\lb{init}.$\\
   \>$\gmoves{\pt{mw}}{\pt{serv}
   }{w_\pt{priv}}{\mathit{Offload} }$
   ${\gto{mw}{cl}\{\lb{reply}\langle ans \rangle.\gend\}} \}  \}\}$
 \end{tabbing}}
By considering a first-class multiparty domain migration primitive at
the type and process levels, we
can specify that the {\em offload} portion of the protocol takes place
after the middleware and the server {\em migrate} to a private
domain $w_\pt{priv}$, as well as ensuring that only
accessible domains can be interacted with.
For instance, the type for
the server that is mechanically {\em projected} from the protocol
above ensures that the server first migrates to the private domain,
communicates with the middleware, and then migrates back to its
initial domain.

Perhaps surprisingly, our {domain-aware} \emph{multiparty} sessions are
studied within a context of logical \emph{binary} domain-aware
sessions, arising from a propositions-as-types interpretation of
hybrid linear logic
\cite{DBLP:journals/entcs/DespeyrouxOP17,chaudhuri:hal-01968154},
with strong {static} correctness guarantees derived from the
logical nature of the system. Multiparty domain-awareness arises
through an interpretation of 
multiparty protocols as \emph{medium processes}~\cite{DBLP:conf/forte/CairesP16}
that orchestrate the
multiparty interaction 
while enforcing the necessary domain-level constraints and migration steps.


\vspace{-0.3cm}
\subparagraph{Contributions}
The key contributions of this work are:
\begin{enumerate}
\item  A process model with explicit domain-based migration (\secref{sec:proc-model}).
We present a session $\pi$-calculus with domains
that can be communicated via novel domain movement prefixes.

\item  
  A session type discipline for domain-aware
  interacting processes~(\secref{sec:hill}).  Building upon an extension
    of linear logic with features from \emph{hybrid
    logic}~\cite{DBLP:journals/entcs/DespeyrouxOP17,chaudhuri:hal-01968154} we generalize the Curry-Howard
  interpretation of session types~\cite{DBLP:conf/concur/CairesP10,DBLP:journals/mscs/CairesPT16} by
  interpreting \emph{(modal) worlds} as \emph{domains} where session
  behavior resides.  
  In our system, types can specify domain \emph{migration} and
  \emph{communication}; domain mobility is governed
  by a parametric accessibility relation.
   Judgments stipulate
  the services  used and realized by processes \emph{and} the domains where
  sessions should be present.
  Our type discipline statically enforces session fidelity,
  global progress and, notably, that communication can only happen
  between accessible domains.

\item 
\jpnote{As a specific application, we introduce a framework}
of domain-aware multiparty sessions~(\secref{sec:hmpst}) \jpnote{that uniformly extends} the standard multiparty session
  framework of \cite{DBLP:conf/popl/HondaYC08} with domain-aware
  migration and communication primitives.
  Our development leverages our
  logically motivated domain-aware {\em binary} sessions (\secref{sec:hill}) to give a
  precise semantics to multiparty sessions through a   (typed)
  \emph{medium process} that acts as an orchestrator of domain-aware
  multiparty interactions, \jpnote{lifting} the strong correctness properties
  of typed processes \jpnote{to the} multiparty setting. 
  We show that 
  mediums soundly 
  and
  completely 
  encode the local behaviors of
  participants in a domain-aware multiparty session.

\end{enumerate}

\noindent
We conclude with a discussion of related work (\secref{s:relw}) and concluding remarks
(\secref{sect:concl}). 
 Appendix~\ref{app:app} lists omitted definitions
 and proofs. We point the interested reader to
 Appendices~\ref{sec:examples} and~\ref{app:examples} for extended
 examples on domain-aware multiparty and binary sessions, respectively.

\section{Process Model}\label{sec:proc-model}
We introduce a 
synchronous
$\pi$-calculus~\cite{sangiorgi-walker:book} with labeled choice and
explicit domain migration and communication.  
We write
$\omega,\omega',\omega''$ to stand for a concrete domain ($w, w', \ldots$) or a domain
variable ($\alpha, \alpha', \ldots$). 
Domains are handled at a high-level of abstraction, with
their identities being attached to session channels. Just as the
$\pi$-calculus allows for communication over names and name mobility,
our model also allows for domain communication and mobility. These
features are justified with the typing discipline of \secref{sec:hill}.

\begin{definition}
\label{d:procs}
Given infinite, disjoint sets $\Names$  of {\em names} $(x,y,z,u,v)$,
$\mathcal{L}$ of labels $l_1,l_2,\dots$,
$\mathcal{W}$ of \emph{domain tags} $(\wtag, \wtag', \wtag'')$ and $\mathcal{V}$ of
{\em domain variables} $(\alpha,\beta,\gamma)$, respectively,  
the set of {\em processes} $(P,Q,R)$ is  defined by
%
%

\vspace{-0.4cm}
{\small
\[ 
\begin{array}{rllllllllllllllllllllllllll}
P  & ::=  & 
 \zero  & \!\!\mid\!\! &   P \para Q   & \!\!\mid\!\! &  (\nub y)P &
     \!\!\mid\!\!  &x\out{y}.P &\!\!\mid\!\! & x(y).P     &  \!\!\mid\!\! & \bang x(y).P  &    &
     \\
    &\!\!\mid\!\! &  \linkr{x}{y} & \!\!\mid\!\! & \mycasebig{x}{\lb{l}_i : P_i}{i \in I}& \!\!\mid\!\! & \mysel{x}{\lb{l}_i};P \\
    & \!\!\mid\!\! &  x\langle y @ \omega\rangle.P
                               & \!\!\mid\!\! & x(y@\omega).P  &
                                                                 \!\!\mid\!\!  & x\langle \omega \rangle.P
  &\!\!\mid\!\! &  x(\alpha).P
\end{array}
\]}
\end{definition}

Domain-aware prefixes are present only in the last line.
As we make precise in the typed setting of \secref{sec:hill}, 
these constructs
realize mobility and domain communication, in the usual sense of the $\pi$-calculus:
migration to a domain is always associated to mobility with a fresh name.

The operators $\zero$ (inaction), $P \para Q$ (parallel composition) 
and $(\nub y)P$ (name restriction) are standard.
We then have $x\out{y}.P $ (send  $y$ on $x$ and
proceed as $P$), $\inp{x}{y}P $ (receive $z$ on $x$ and proceed
as $P$ with parameter $y$ replaced by $z$), and $ \bang
x(y).P$ which denotes replicated (persistent) input. 
The forwarding construct $\linkr{x}{y}$ equates 
$x$ and $y$; 
it is a primitive representation of a copycat process.
The last two constructs in the second line define a labeled choice
mechanism:
$\mycasebig{x}{\lb{l}_i : P_i}{i \in I}$ is a process that awaits
some label $l_j$ (with $j\in I$) and proceeds as
$P_j$. Dually, the process $\mysel{x}{\lb{l}_i};P$ emits a label
$\lb{l}_i$ and proceeds as $P$.

The first
two operators in the third line define explicit domain migration:
given a domain $\omega$, $x\langle y @ \omega \rangle.P$ denotes a
process that is prepared to migrate the communication actions in $P$
on endpoint $x$, to session $y$ on 
$\omega$.  Complementarily, process $x(y @ \omega).P$ signals an
endpoint $x$ to move to $\omega$, providing $P$ with the appropriate
session endpoint that is then bound to $y$. In a typed setting, domain
movement will be always associated with a fresh session channel.
Alternatively, this form of coordinated migration can be read as an
explicit form of agreement (or authentication) in trusted domains.
Finally, the last two operators in the third line define output and input of
domains, $x\langle \omega\rangle.P$ and $x(\alpha).P$, respectively.
These constructs allow for domain information to be obtained and propagated across
processes dynamically.

Following~\cite{DBLP:journals/tcs/Sangiorgi96a}, we abbreviate 
$(\nub y_{})x\out {y}$
and $(\nub y_{})x\out {y @ \omega}$
as $\outp{x}{y}$ and $\outp{x}{y@ \omega}$, respectively.
In 
$(\nub y_{})P$,
$x_{}(y_{}).P$, and $x_{}(y_{} @ \omega).P$ the distinguished
occurrence of name $y_{}$ is binding with scope  $P$. Similarly for
$\alpha$ in $x(\alpha).P$.
We identify processes
up to consistent renaming of bound names and variables, writing $\equiv_{\alpha}$
for this congruence.
$P\subst{x_{}}{y_{}}$ denotes 
the 
capture-avoiding
substitution of  $x_{}$ for  $y_{}$ in $P$.
%
While \emph{structural congruence}  $\equiv$  expresses
standard identities on the 
basic 
structure of
processes (cf. \cite{longversion}), \emph{reduction} 
expresses their behavior. 

%

  {\em Reduction} ($P\red Q$) is the binary relation defined by the 
 rules below \jpnote{and closed under structural congruence}; it specifies the computations that a process performs on its own.
  \[\begin{array}{ll}
\send{x}{y}.Q\para \inp{x}{z}{P} \red Q\para P\subst{y}{z_{}} &
\send{x}{y}.Q\para \bang \inp{x}{z_{}}{P} \red 
                                                           Q\para P\subst{y}{z_{}} \para \bang \inp{x}{z_{}}{P} \\
      x\langle y @ \omega\rangle.P \mid x(z@\omega').Q \red P \mid Q\{y/z\} &
       x\langle\omega\rangle.P \mid x(\alpha).Q \red P \mid Q\{\omega/\alpha\}\\                                                                       
(\nub x_{})(\linkr{x}{y} \para P) \red P\subst{y}{x} & 
Q\red Q'\Rightarrow P \para Q \red P \para Q'   \\
P\red Q\Rightarrow(\nub y)P\red(\nub y)Q  &
\mysel{x}{\lb{l}_j};P \para \mycasebig{x}{\lb{l}_i : Q_i}{i \in I}      \red P \para Q_j \quad (j \in I)
\end{array}
\]
\noindent For the sake of generality, 
reduction 
allows dual endpoints with the same name 
to interact,
independently 
of the domains of their subjects. 
The type system introduced next
will ensure, among other things, \emph{local reductions}, disallowing
synchronisations among distinct domains.





\section{Domain-aware Session Types via Hybrid Logic}\label{sec:hill}

This section develops a new domain-aware formulation of binary session
types. Our system is based on a Curry-Howard interpretation of a
linear variant of so-called \emph{hybrid logic}, and can be seen as an
extension of the interpretation of~\cite{DBLP:conf/concur/CairesP10,DBLP:journals/mscs/CairesPT16}
to hybrid (linear) logic.
Hybrid logic is often used as an umbrella term for a class of logics
that extend the expressiveness of propositional logic by considering
modal \emph{worlds} as syntactic objects that occur in propositions.

As in~\cite{DBLP:conf/concur/CairesP10,DBLP:journals/mscs/CairesPT16}, propositions are
interpreted as session types of communication channels, proofs as
typing derivations, and proof reduction as process
communication. As main novelties, here  we interpret:
\btnote{logical} worlds as \emph{domains}; 
the hybrid connective $@_\omega \,A$ as the type of
a session that {\em migrates} to an accessible domain $\omega$; 
and type-level quantification over worlds $\forall \alpha . A $ and $\exists
\alpha. A$ as {\em \btnote{domain} communication}. We also consider a type-level
operator $\here{\alpha.A}$ (read ``here'') which binds the {\em current} domain of the
session to $\alpha$ in $A$.
The syntax of domain-aware session types is given in
\defref{d:hillprops}, where $w, w_1, \ldots$ stand for \btnote{domains}
drawn from $\mathcal{W}$, and where
$\alpha, \beta $ and $\omega, \omega'$ are used
as in the syntax of processes.

\begin{definition}[Domain-aware Session Types] \label{d:hillprops}
The syntax of types $(A,B,C)$ 
is  defined by

\vspace{-0.4cm}
{\small\[
 \begin{array}{rllllllllllllllllllllllllll}
A  & ::= & \one &  \!\!   |\!\!  & A\lolli B &   \!\!  |\!\!  &  A \tensor B  &
       \!\!  |\!\!  &  \with\{\lb{l}_i :
                                               A_i\}_{i\in I}  & \!\! |\!\! &  \oplus\{\lb{l}_i :
                                               A_i\}_{i\in I} & \!\!|\!\! &
                                                                            \bang A  \\
   & |\!\!  & @_{\omega}\,A  
                & \!\!|\!\!  &  \forall \alpha . A &\!\!|\!\!  & \exists \alpha . A &
                  \!\!|\!\!  & \here{\alpha.A}                                                                    
\end{array}
\]}
\end{definition}

Types are the propositions of intuitionistic linear logic where the
additives $A\with B$ and $A\oplus B$ are generalized to a labelled
$n$-ary variant. Propositions take the standard interpretation as
session types, extended with hybrid logic operators
\cite{braunerdepaiva06}, with \btnote{worlds interpreted as domains} that are \btnote{explicitly}
subject to an {\em accessibility relation} \btnote{(in the style of
  \cite{DBLP:phd/ethos/Simpson94a})} that is tracked by environment
$\Omega$. Intuitively, $\Om$ \btnote{is made up of} direct
accessibility hypotheses of the form $\omega_1 \prec \omega_2$, meaning
that domain $\omega_2$ is accessible from $\omega_1$.

Types are assigned to channel names;
a \emph{type assignment} $x{:}A[\omega]$ 
enforces the use of name $x$ according to session $A$, \emph{in the domain $\omega$}. 
A \emph{type environment} is a collection of type assignments.
Besides the 
accessibility environment $\Om$ just mentioned, our typing judgments consider 
two kinds of type environments: a
\emph{linear} part $\Delta$ and an \emph{unrestricted} part $\Gamma$.
They are subject to different structural properties:
weakening and contraction principles hold for $\Gamma$ but not 
for $\Delta$. Empty environments are written as `$\,\cdot\,$'.
We then consider two judgments:

\vspace{-0.6cm}
{\[
\begin{array}{lll}
\text{(i)}~~ \Omega \vdash \omega_1 \prec \omega_2 & \quad\mbox{ and } \quad&
                                                                    \text{(ii)}~~ \Omega ; \G ; \D \vdash P :: z{:} A[\omega]
\end{array} 
\]}
Judgment $\text{(i)}$ states that $\omega_1$ can directly access $\omega_2$ under
the hypotheses in $\Omega$.   We write $\prec^*$ for the {reflexive, transitive
  closure} of $\prec$, and $\omega_1 \not\prec^* \omega_2$ when
$\omega_1 \prec^* \omega_2$ does not hold.  Judgment~\text{(ii)} states that
process $P$ offers the session behavior specified by type $A$ on
channel~$z$; the session $s$ resides at domain $\omega$, under the
accessibility hypotheses $\Omega$, using  {unrestricted}
sessions in $\G$ and {linear} sessions in $\D$.  
Note that each hypothesis in $\G$ and $\D$ is 
labeled with a specific domain.  We omit $\Omega$ when it is clear from
context. 

\subparagraph{Typing Rules}
\jpnote{Selected} typing rules are given in
Fig.~\ref{fig:phillrules}; \btnote{see \cite{longversion} for the
  full listing.}
Right rules (marked with $\rgt{}$) 
specify how to \emph{offer} a session of a
given type, left rules (marked with $\lft{}$)  define how to \emph{use} a session. 
The hybrid nature of the system induces a notion of \emph{well-formedness} of sequents: 
a sequent 
$\Omega ; \G ; \D \vdash P :: z: C[\omega_1]$ 
is
\emph{well-formed} if 
$\Omega \vdash \omega_1\,\prec^*\,\omega_2$ for every $x{:}A[\omega_2] \in \D$, which we
abbreviate as $\Omega \vdash \omega_1 \prec^* \D$, meaning that all domains
mentioned in $\D$ are accessible  from $\omega_1$ (not necessarily in a single
\emph{direct} step).
No such 
domain 
requirement is imposed on $\G$.  
If an end sequent is well-formed, every sequent in its proof will
also be well-formed.  
All rules (read bottom-up)
preserve this invariant; only \name{$\cut$}, \name{$\cpy$}, \name{$\rgt@$},
\name{$\lft \forall$} and \name{$\rgt \exists$}
require explicit  checks, which we discuss below.
This invariant 
 statically excludes interaction between sessions in 
accessible domains (cf.~Theorem~\ref{thrm:dom}).

We briefly discuss some of the typing rules, first noting that we
consider processes modulo structural congruence; hence,
typability is closed under $\equiv$ by definition.
Type $A\lolli B$ denotes a session that inputs a session of type $A$
and proceeds as $B$.
To offer $z{:}A\lolli  B$ at domain $\omega$, we input $y$ along $z$
that will offer $A$ at $\omega$ and proceed, now offering
$z{:}B$ at $\omega$:

\vspace{-0.3cm}
{\small\[
  \inferrule*[left=\name{$\rgt\lolli$}]
{\bluee{\Omega ; \G ; \D ,\,}\redd{ y{:}}\bluee{A[\omega]} \vdash
  \redd{P :: z{:}} \bluee{B[\omega]}}{\bluee{\Omega ; \G ; \D} \vdash
  \redd{z(y).P :: z{:}}\bluee{ A\lolli B[\omega]}}
\quad
\inferrule*[left=\name{$\rgt\tensor$}]{\bluee{\Omega ; \G ; \D_1} \vdash \redd{P :: y{:}}\bluee{A[\omega] \quad \Omega ; \G ; \D_2} \vdash \redd{Q:: z{:}}\bluee{B[\omega]}}{\bluee{\Omega ; \G ; \D_1 , \D_2} \vdash \redd{\outp{z}{y}.(P \mid  Q) :: z{:}} \bluee{A \tensor B [\omega]}}
\]} Dually, $A\tensor B$ denotes a session that outputs a session that
will offer $A$ and continue as~$B$.  To offer $z{:}A\otimes B$, we
  output  a fresh name $y$ with type $A$ along $z$
 and  proceed offering $z{:}B$.

\jpnote{The \name{$\cut$} rule allows us to compose process $P$, which offers
$x{:}A[\omega_2]$, with process $Q$, which uses $x{:}A[\omega_2]$ to offer
$z{:}C[\omega_1]$. 
We require that {domain} $\omega_2$ is accessible from $\omega_1$
(i.e., $\omega_1 \prec^* \omega_2$). 
We also require $\omega_1 \prec^* \D_1$: the domains mentioned in $\Delta_1$ (the context for $P$) must be accessible from $\omega_1$, 
which follows from the transitive closure
of the accessibility relation ($\prec^*$) using the
intermediary domain $\omega_2$. As in~\cite{DBLP:conf/concur/CairesP10,DBLP:journals/mscs/CairesPT16}, composition binds the name $x$:}

\vspace{-0.4cm}
{\small\[
\inferrule*[left=\name{$\cut$}]
{
    \Omega \vdash \omega_1 \prec^* \omega_2 
    \quad
   \Omega \vdash
   \omega_1 \prec^* \D_1  
   \quad \bluee{\Omega ; \G ; \D_1} \vdash \redd{P :: x{:}} \bluee{A[\omega_2]
    \quad 
 \Omega ; \G ; \D_2 ,\,} \redd{x{:}}\bluee{A[\omega_2]} \vdash \redd{Q ::
    z{:}}\bluee{C[\omega_1]}
    }
{\bluee{\Omega ; \G ; \D_1 , \D_2}  \vdash \redd{(\nub x)(P \mid Q) :: z{:}}\bluee{C[\omega_1]}}
\]}
    Type $\one$ means that no further interaction will
    take place on the session; names of type $\one$ may be passed
    around as opaque values.  $\with\{\lb{l}_i : A_i\}_{i\in I}$
    types a session channel that offers its partner a choice between
    the  $A_i$ behaviors, each uniquely identified by a label $\lb{l}_i$. 
    Dually, $\oplus\{\lb{l}_i : A_i\}_{i\in I}$ types a session
    that selects some behavior $A_i$ by emitting the corresponding label.
    \jpnote{For flexibility and consistency with
    merge-based projectability in multiparty session types, rules for choice and selection induce a standard notion of session
    subtyping \cite{DBLP:journals/acta/GayH05}.}
    
Type $\bang A$ types a shared (non-linear) channel, 
to be used by a server for spawning an arbitrary number of new sessions
(possibly none), each one conforming to type $A$.

Following our previous remark on well-formed sequents,
the only rules 
that appeal to accessibility are 
\name{$\rgt @$}, \name{$\lft @$}, \name{$\cpy$}, and \name{$\cut$}.
These conditions are directly associated with varying
degrees of flexibility in terms of typability, depending on what
relationship is imposed between the \btnote{domain} to the left and to the right
of the turnstile in the left rules.  Notably, our system leverages the
accessibility judgment to enforce that communication is only allowed
between processes whose sessions are in (transitively) {\em accessible}
domains.

The type operator $@_\omega$ realizes a \emph{domain migration} mechanism which 
is specified both at the level of types 
and processes via name mobility tagged with a domain name. 
Thus, a channel typed with $@_{\omega_2} A$ denotes that behavior $A$ is
available by first \emph{moving to} domain $\omega_2$,
directly accessible from the current domain.
    More precisely, we have:

    \vspace{-0.5cm}
    {\small\[
    \inferrule*[left=\name{$\rgt @$}]
    {\begin{array}{c}
      \bluee{\Omega} \vdash \bluee{\omega_1 \prec \omega_2}
\\ \bluee{\Omega} \vdash \bluee{\omega_2 \prec^* \D} \quad \Omega ;
      \G ; \D \vdash \redd{P    :: y{:}} \bluee{A[\omega_2]}
      \end{array}}
{\bluee{\Omega ; \G ;\D} \vdash \redd{\overline{z}\langle y @
    \omega_2\rangle.P :: z{:}} \bluee{@_{\omega_2} A[\omega_1]}}
    \quad
    \inferrule*[left=\name{$\lft @$}]
{\begin{array}{c}\Omega , \omega_2 \prec \omega_3 ; \G ; \D ,\, \redd{ y{:}}A[\omega_3] \vdash \redd{P :: z{:}} C[\omega_1]\end{array}}{\bluee{\Omega ; \G ; \D ,\,} \redd{x{:}}\bluee{@_{\omega_3} A[\omega_2]} \vdash \redd{x(y@\omega_3).P :: z{:}}\bluee{ C[\omega_1]}}
    \]}
    Hence, 
    a process \emph{offering} a behavior $z{:}@_{\omega_2}\, A$ at $\omega_1$ ensures: 
    (i)~behavior $A$ is available at $\omega_2$
    along a \emph{fresh} session channel $y$ that is emitted along $z$
    and (ii)~$\omega_2$ is directly accessible from $\omega_1$.  To maintain
    well-formedness of the sequent we also must check that all
    domains in $\Delta$ are still accessible from $\omega_2$.
    Dually, \emph{using} a service $x{:}@_{\omega_3} A [\omega_2]$ 
    entails receiving a channel $y$ that will offer behavior $A$ at
    domain $\omega_3$ (and also allowing the usage of the fact that $\omega_2
    \prec \omega_3$).



    Domain-quantified sessions introduce
    domains as \emph{fresh} parameters to types: a particular service can be
    specified with the ability to refer to any
    existing directly accessible domain (via universal
    quantification) or to some \emph{a priori} unspecified accessible
    domain:

\vspace{-0.5cm}
{\small\[
\inferrule*[left=\name{$\rgt \forall$}]
{\bluee{\Omega , \omega_1 \prec \alpha ; \G ; \D} \vdash \redd{P ::
    z{:}}\bluee{A[\omega_1]} \quad \alpha \not\in \Omega , \G , \D , \omega_1}
{\bluee{\Omega ; \G ; \D} \vdash \redd{z(\alpha).P :: z{:}}\bluee{\forall \alpha . A[\omega_1]}}
~~~
\inferrule*[left=\name{$\lft \forall$}]{\begin{array}{c} 
\Omega \vdash \omega_2 \prec \omega_3 \\
 \Omega ; \G ; \D ,\, \redd{x{:}}A\{\omega_3 / \alpha\}  [\omega_2] \vdash \redd{Q :: z{:}}C[\omega_1]
\end{array}}{\bluee{\Omega ; \G ; \D ,\,} \redd{x{:}}\bluee{\forall \alpha . A[\omega_2]} \vdash
  \redd{x\out{\omega_3}.Q :: z{:}}\bluee{C[\omega_1]}}
\]}
Rule \name{$\rgt \forall$} states that
    a process seeking to offer $\forall \alpha .A[\omega_1]$ denotes a
    service that is located at domain $\omega_1$ but that may refer to any
    fresh domain directly accessible from $\omega_1$ in its
    specification (e.g.~through the use of $@$). Operationally, this means that
    the process must be ready to receive from its client a reference to
    the domain being referred to in the type, which is bound to $\alpha$
    (occurring fresh in the typing derivation). Dually,
    Rule~\name{$\lft \forall$} indicates that a process interacting with a
    service of type $x{:}\forall \alpha .A[\omega_2]$ must make concrete
    the domain that is directly accessible from $\omega_2$ it wishes to
    use, which is achieved by the appropriate output action.  Rules
    {\name{$\lft \exists$}} and {\name{$\rgt \exists$}} for the
    existential quantifier have a dual reading.

Finally, the type-level operator $\here{\alpha.A}$ 
allows for a type to refer to its {\em current} domain:

\vspace{-0.4cm}
{\small\[
\inferrule*[left=\name{$\rgt\downarrow$}]
  {\Omega ; \G ; \D \vdash \redd{P :: z{:}}\bluee{A\{\omega/\alpha\} [\omega]}}
  {\Omega ; \G ; \D \vdash \redd{P :: z{:}}\bluee{\here{\alpha . A} [\omega]}}
  \qquad
  \inferrule*[left=\name{$\lft\downarrow$}]
   {\Omega ; \G ; \D , \redd{x{:}}\bluee{A\{\omega/\alpha\}[\omega]} \vdash
  \redd{P ::  z{:}}\bluee{C}}
  {\Omega ; \G ; \D , \redd{x{:}}\bluee{\here{\alpha.A}[\omega]} \vdash
  \redd{P ::  z{:}}\bluee{C}}
\]}
The typing rules that govern $\here{\alpha.A}$ are completely
symmetric and produce no action at the process level, merely
instantiating the \btnote{domain} variable $\alpha$ with the current domain
$\omega$ of the session. As will be made clear in \secref{sec:hmpst},
this connective plays a crucial role in ensuring the correctness of
our analysis of multiparty domain-aware sessions in our logical setting.

By developing our type theory with an \btnote{explicit} domain accessibility
judgment, we can consider the accessibility relation as a
\emph{parameter} of the framework. 
This allows changing accessibility
relations and \btnote{their} properties without having to alter the entire
system.
To consider the simplest possible
accessibility relation, the only defining rule for accessibility would
be  
Rule \name{$\mathsf{whyp}$} in 
Fig.~\ref{fig:phillrules}.
To consider an accessibility relation which is an 
equivalence relation we would add reflexivity, transitivity, and
symmetry rules to the judgment.

 \vspace{-0.2cm}
\subparagraph{Discussion and Examples}
\jpnote{Being an interpretation of {\em hybridized} linear logic, 
our domain-aware theory is \emph{conservative} wrt the Curry-Howard interpretation of session types
in~\cite{DBLP:conf/concur/CairesP10,DBLP:journals/mscs/CairesPT16}, in the following sense: 
the system in~\cite{DBLP:conf/concur/CairesP10,DBLP:journals/mscs/CairesPT16} 
corresponds to the case where every session resides at the same domain. 
As in~\cite{DBLP:conf/concur/CairesP10,DBLP:journals/mscs/CairesPT16}, 
the sequent calculus for the underlying (hybrid) linear logic can be recovered 
from our typing rules
by erasing processes and name assignments.}



Conversely, a fundamental consequence of our hybrid interpretation is
that it {\em refines} the session type structure in non-trivial
ways. By requiring that communication only \jpnote{occurs} between sessions
located at the same (or accessible) domain 
    we effectively introduce a new layer of reasoning
to session type systems. To illustrate this feature, consider the
following session type $\mathsf{WStore}$, which specifies a simple interaction between a web store
and its clients:

    \vspace{-0.5cm}
    {\small\[
\begin{array}{l}
\!\!\mathsf{WStore}   \triangleq  \mathsf{addCart} \lolli \with\!\{\lb{buy}: \mathsf{Pay} \, , \,
  \lb{quit}: \one \} \qquad
  \!\!\mathsf{Pay}   \triangleq  
\mathsf{CCNum} \lolli \oplus\!\{\lb{ok} : \mathsf{Rcpt}
\tensor \one \, , \, \lb{nok} : \one\}
 \end{array}
\]}
$\mathsf{WStore}$ allows clients to checkout their shopping carts by emitting a $\lb{buy}$
message or to $\lb{quit}$. In the former case, the client pays for the
    purchase by sending their credit card data. If a banking service (not shown) approves the transaction (via an $\lb{ok}$ message), a receipt is emitted. 
Representable in existing session type
systems
(e.g. \cite{DBLP:conf/concur/CairesP10,DBLP:journals/jfp/Wadler14,DBLP:conf/esop/HondaVK98}), types $\mathsf{WStore}$ and $\mathsf{Pay}$  describe the intended communications  but fail to capture the
    crucial fact that in practice the client's sensitive
    information should only be requested after entering a secure
    domain.
To address this limitation, we can use  type-level domain migration 
 to \emph{refine} $\mathsf{WStore}$ and $\mathsf{Pay}$:


    \vspace{-0.4cm}
    
    {\small\[
  \begin{array}{rl}
\mathsf{WStore}_{\mathtt{sec}} \triangleq & \mathsf{addCart} \lolli
\with \{ \lb{buy}: @_{\mathtt{sec}}\,\mathsf{Pay}_{\mathtt{bnk}},
\mathtt{quit}: \one \}
\\
\mathsf{Pay}_{\mathtt{bnk}}  \triangleq  & \mathsf{CCNum} \lolli \oplus\!\{\lb{ok} : (@_{\mathtt{bnk}}\mathsf{Rcpt})
    \tensor \one  ,  \lb{nok} : \one\}
    \end{array}
    \]
}
$\mathsf{WStore}_{\mathtt{sec}}$   decrees that the interactions
pertinent to 
type $\mathsf{Pay}_{\mathtt{bnk}}$ should be preceded by a
migration step to the trusted domain $\mathtt{sec}$, which should be
directly accessible from $\mathsf{WStore}_{\mathtt{sec}}$'s current
    domain.
The type  also specifies that the receipt
must originate from a bank domain $\mathtt{bnk}$ (e.g., ensuring that the receipt is never produced by the store without entering $\mathtt{bnk}$).
When considering the interactions with a client (at domain
    $\mathtt{c}$) that checks out their cart, we reach a state that is
    typed with the following judgment:
    \vspace{-0.2cm}

    {\small
    \[
\mathtt{c}\prec\mathtt{ws} ; \cdot ;
x{:}@_{\mathtt{sec}}\mathsf{Pay}_{\mathtt{bnk}}[\mathtt{ws}] \vdash Client ::
    z{:}@_{\mathtt{sec}}\one[\mathtt{c}]
\]}
\noindent 
At this point, it is \emph{impossible} for a (typed) client to interact with the
behavior that is protected by the domain $\mathtt{sec}$, since
it is not the case that $\mathtt{c}\prec^* \mathtt{sec}$. \btnote{That
    is, no judgment of the form $\mathtt{c}\prec\mathtt{ws} ; \cdot ;
    \mathsf{Pay}_\mathtt{bnk}[\mathtt{sec}] \vdash
    Client' :: z{:}T[\mathtt{c}]$
    is derivable.
    }
This ensures,
e.g., that a client cannot exploit the payment platform of the web
store by accessing the trusted domain in unforeseen ways.
%
The client can only communicate in the secure domain \emph{after} the web
store service has migrated accordingly, as shown by the judgment     

\vspace{-0.4cm}
    {\small
    \[
    \mathtt{c}\prec\mathtt{ws} , \mathtt{ws}\prec \mathtt{sec} ; \cdot ;
x'{:}\mathsf{Pay}_{\mathtt{bnk}}[\mathtt{sec}] \vdash Client' :: z'{:}\one[\mathtt{sec}].
\]}
\vspace{-0.2cm}

\begin{figure}[!t]
  {\footnotesize
    \[
\begin{array}{ccccc}
\inferrule*[left=\name{$\mathsf{whyp}$}]{ }{\bluee{\Omega , \omega_1 \prec \omega_2 \vdash \omega_1 \prec \omega_2}}
\qquad
\inferrule*[left=\name{$\mathsf{id}$}]{ }{\bluee{\Omega ; \G ;\,} \redd{ x{:}}\bluee{A[\omega]} \vdash \redd{\linkr{x}{z} :: z{:}}\bluee{A[\omega]}}
\\[1em]
\inferrule[\name{$\rgt @$}]{\bluee{\Omega} \vdash \bluee{\omega_1 \prec \omega_2
\quad \bluee{\Omega} \vdash \bluee{\omega_2 \prec^* \D} \quad \Omega ; \G ; \D} \vdash \redd{P    :: y{:}} \bluee{A[\omega_2]}}
{\bluee{\Omega ; \G ;\D} \vdash \redd{\overline{z}\langle y @ \omega_2\rangle.P :: z{:}} \bluee{@_{\omega_2} A[\omega_1]}}
\quad
\inferrule[\name{$\lft @$}]
{\begin{array}{c}\Omega , \omega_2 \prec \omega_3 ; \G ; \D ,\, \redd{ y{:}}A[\omega_3] \vdash \redd{P :: z{:}} C[\omega_1]\end{array}}{\bluee{\Omega ; \G ; \D ,\,} \redd{x{:}}\bluee{@_{\omega_3} A[\omega_2]} \vdash \redd{x(y@\omega_3).P :: z{:}}\bluee{ C[\omega_1]}}
  \\[1em]
\inferrule[\name{$\rgt \forall$}]
{\bluee{\Omega , \omega_1 \prec \alpha ; \G ; \D} \vdash \redd{P ::
    z{:}}\bluee{A[\omega_1]} \quad \alpha \not\in \Omega , \G , \D , \omega_1}
{\bluee{\Omega ; \G ; \D} \vdash \redd{z(\alpha).P :: z{:}}\bluee{\forall \alpha . A[\omega_1]}}
\quad
\inferrule[\name{$\lft \forall$}]{
\Omega \vdash \omega_2 \prec \omega_3 \quad
 \Omega ; \G ; \D ,\, \redd{x{:}}A\{\omega_3 / \alpha\}  [\omega_2] \vdash \redd{Q :: z{:}}C[\omega_1]
}{\bluee{\Omega ; \G ; \D ,\,} \redd{x{:}}\bluee{\forall \alpha . A[\omega_2]} \vdash
  \redd{x\out{\omega_3}.Q :: z{:}}\bluee{C[\omega_1]}} 
  \\[1.5em]
  \inferrule*[left=\name{$\rgt \exists$}]
{ \bluee{\Omega} \vdash \omega_1 \prec \omega_2 \quad \bluee{\Omega ; \G ; \D} \vdash \redd{P ::
    z{:}}\bluee{A\{\omega_2/\alpha\}[\omega_1]} }
{\bluee{\Omega ; \G ; \D} \vdash \redd{z\langle\omega_2\rangle.P :: z{:}}\bluee{\exists \alpha . A[\omega_1]}}
\quad
\inferrule*[left=\name{$\lft \exists$}]{\begin{array}{c} 
 \Omega , \omega_2 \prec \alpha; \G ; \D ,\, \redd{x{:}}A[\omega_2] \vdash \redd{Q :: z{:}}C[\omega_1]
\end{array}}{\bluee{\Omega ; \G ; \D ,\,} \redd{x{:}}\bluee{\exists \alpha . A[\omega_2]} \vdash
  \redd{x(\alpha).Q :: z{:}}\bluee{C[\omega_1]}} 
  \\[1em]
  \inferrule*[left=\name{$\rgt\downarrow$}]
  {\Omega ; \G ; \D \vdash \redd{P :: z{:}}\bluee{A\{\omega/\alpha\} [\omega]}}
  {\Omega ; \G ; \D \vdash \redd{P :: z{:}}\bluee{\here{\alpha . A} [\omega]}}

  \quad
  \inferrule*[left=\name{$\lft\downarrow$}]
   {\Omega ; \G ; \D , \redd{x{:}}\bluee{A\{\omega/\alpha\}[\omega]} \vdash
  \redd{P ::  z{:}}\bluee{C}}
  {\Omega ; \G ; \D , \redd{x{:}}\bluee{\here{\alpha.A}[\omega]} \vdash
  \redd{P ::  z{:}}\bluee{C}}
  \\[1em]
\inferrule*[left=\name{$\cpy$}]{\bluee{\Omega \vdash \omega_1 \prec^* \omega_2 \quad \Omega ; \G ,\,} \redd{u{:}}\bluee{A[\omega_2] ; \D ,\,} \redd{y{:}}\bluee{A[\omega_2]}  \vdash \redd{P :: z{:}}\bluee{ C[\omega_1]}}{\bluee{\Omega ; \G ,\,}\redd{ u{:}}\bluee{A[\omega_2] ; \D} \vdash \redd{\outp{u}{y}.P  :: z{:}}\bluee{C[\omega_1]}}
\\[1em]
\inferrule*[left=\name{$\cut$}]
{
\Omega \vdash \omega_1 \prec^* \omega_2 \quad \Omega \vdash \omega_2 \prec^* \D_1  
\quad
 \bluee{\Omega ; \G ; \D_1} \vdash \redd{P :: x{:}} \bluee{A[\omega_2]
    \quad 
 \Omega ; \G ; \D_2 ,\,} \redd{x{:}}\bluee{A[\omega_2]} \vdash \redd{Q ::
 z{:}}\bluee{C[\omega_1]}
 }
{\bluee{\Omega ; \G ; \D_1 , \D_2}  \vdash \redd{(\nub x)(P \mid Q) :: z{:}}\bluee{C[\omega_1]}}
\end{array}
\]}
\caption{Typing Rules (Excerpt -- see \btnote{\cite{longversion}})\label{fig:phillrules}}
\vspace{-3ex}
\end{figure}



\vspace{-0.5cm}
\subparagraph{Technical Results}\label{sec:techn}
We state the main results of type safety via type preservation (Theorem~\ref{thrm:pres}) and global progress (Theorem~\ref{thrm:prog}).
These results directly ensure session fidelity and deadlock-freedom.
Typing also ensures termination, 
i.e., processes do not exhibit infinite reduction paths
(Theorem~\ref{thrm:sn}). \btnote{We note that in the presence of
  termination, our progress result ensures that communication actions are
  always guaranteed to take place.}
Moreover, 
as a property specific to domain-aware processes, we show
\emph{domain preservation}, i.e., 
processes 
respect their domain accessibility conditions (Theorem~\ref{thrm:dom}).
The formal development of these results relies on a
\emph{domain-aware} labeled transition system \cite{longversion}, defined as 
a simple
generalization of the early labelled transition system for the session $\pi$-calculus given in~\cite{DBLP:conf/concur/CairesP10,DBLP:journals/mscs/CairesPT16}.


\smallskip

\noindent{\bf Type Safety and Termination.}
%
Following~\cite{DBLP:conf/concur/CairesP10,DBLP:journals/mscs/CairesPT16}, our proof of type
preservation relies on a simulation between reductions in the
session-typed $\pi$-calculus and logical proof reductions.

\begin{lemma}[Domain Substitution]\label{lem:wsubst}
  \jpnote{Suppose $\Omega \vdash \omega_1 \prec \omega_2$. Then we have:}
  {\small\begin{itemize}
    \item If $\Omega , \omega_1 \prec \alpha , \Omega'; \G ; \D
\vdash P :: z{:} A[\omega]$ then  \\ 
$\Omega ,\Omega'\{\omega_2/\alpha\};
\G\{\omega_2/\alpha\}; \D\{\omega_2/\alpha\} \vdash
P\{\omega_2/\alpha\} :: z{:}A[\omega\{\omega_2/\alpha\}]$.
\item $\Omega , \alpha \prec \omega_2 , \Omega'; \G ; \D
\vdash P :: z{:} A[\omega]$ then \\ 
$\Omega ,\Omega'\{\omega_1/\alpha\};
\G\{\omega_1/\alpha\}; \D\{\omega_1/\alpha\} \vdash
P\{\omega_1/\alpha\} :: z{:}A[\omega\{\omega_1/\alpha\}]$.
\end{itemize}}
\end{lemma}

\noindent Safe domain communication relies on
domain substitution preserving typing (Lemma~\ref{lem:wsubst}).


\begin{theorem}[Type Preservation]
\label{thrm:pres}
If $\Omega ; \G ; \D \vdash P :: z{:}A[\omega]$ and $P \tra{} Q$ then \\ $\Omega
; \G ; \D \vdash Q :: z{:}A[\omega]$.
\end{theorem}
\begin{proof}[Proof (Sketch)]
  The proof mirrors those of
  \cite{DBLP:conf/concur/CairesP10,DBLP:journals/mscs/CairesPT16,DBLP:conf/esop/CairesPPT13,Toninho2011},
  relying on a series of lemmas  
  relating the result of dual process actions (via our LTS semantics)
  with typable parallel compositions through the \name{$\cut$}
  rule \btnote{\cite{longversion}}. For session type constructors of
  \cite{DBLP:conf/concur/CairesP10}, the results are unchanged.
   For the domain-aware session type constructors, the
  development is identical that of \cite{DBLP:conf/esop/CairesPPT13}
  and \cite{Toninho2011}, which deal with communication of types and
  data terms, respectively.
\end{proof}

Following~\cite{DBLP:conf/concur/CairesP10,DBLP:journals/mscs/CairesPT16}, the proof of global progress
relies
on a notion of a \emph{live} process,
which intuitively consists of a process that has not yet fully carried
out its ascribed session behavior, and thus is a parallel composition
of processes where at least one is a non-replicated process, guarded
by some action. Formally, we define $live(P)$ if and only if $P \equiv
(\nub \tilde{n})(\pi.Q \mid R)$, for some $R$, 
names $\tilde{n}$ and a non-replicated guarded process $\pi.Q$.

\begin{theorem}[Global Progress]
\label{thrm:prog}
If $\Omega ; \cdot ; \cdot \vdash P :: x{:}\one[\omega]$ and $live(P)$ then
$\exists Q$ s.t. $P\tra{} Q$.
\end{theorem}
Note that 
Theorem~\ref{thrm:prog} is without
loss of generality since using the cut rules we can compose arbitrary well-typed processes
together and $x$ need not occur in $P$ due to Rule
\name{$\rgt\one$}.
  
Termination (strong normalization) is a relevant property for
interactive systems: while from a global perspective they are meant to
run forever, at a local level participants should always react within
a finite amount of time, and never engage into infinite internal
behavior.  We say that a process $P$ \emph{terminates}, noted
$P \Downarrow$, if there is no infinite reduction path from $P$.
  
  \begin{theorem}[Termination]
  \label{thrm:sn}
  If $\Omega ; \G ; \D \vdash P :: x{:}A[\omega]$ then $P \Downarrow$.
  \end{theorem}
  \begin{proof}[Proof (Sketch)]
    By adapting the \emph{linear} logical relations given
    in~\cite{DBLP:conf/esop/PerezCPT12,DBLP:journals/iandc/PerezCPT14,DBLP:conf/esop/CairesPPT13}.
    For the system in \secref{sec:hill} without quantifiers, the
    logical relations correspond to those
    in~\cite{DBLP:conf/esop/PerezCPT12,DBLP:journals/iandc/PerezCPT14}, extended to carry
    over 
    $\Omega$. When considering quantifiers, the logical relations
    resemble those proposed for polymorphic session types
    in~\cite{DBLP:conf/esop/CairesPPT13}, noting that
    no impredicativity concerns are
    involved. 
  \end{proof}
  
  \noindent{\bf Domain  Preservation.}
As a consequence of the hybrid nature of our system, well-typed
processes are guaranteed not only to faithfully 
perform their prescribed behavior
in a deadlock-free manner, but they also do so without
breaking the constraints put in place on domain accessibility given by
our well-formedness constraint on sequents. 

\begin{theorem}~\label{thrm:NI}
Let $\mathcal E$ be a derivation of $\Omega ; \G ; \D \vdash P ::
z{:}A[\omega]$. If $\Omega ; \G ; \D \vdash P :: z{:}A[\omega]$ is well-formed
then every sub-derivation in $\mathcal E$ well-formed.
\end{theorem}

While inaccessible domains
 can appear in $\Gamma$,
such channels can never be used and thus can not appear in a
well-typed process due to the restriction on the $(\mathsf{copy})$ rule.
Combining Theorems~\ref{thrm:pres} and \ref{thrm:NI} we can then show
that even if \jpnote{a session in the environment} changes domains, typing ensures that such a domain will be
(transitively) accessible:

\begin{theorem}~\label{thrm:dom}
Let 
(1) $\Omega ; \G ; \D , \D' \vdash (\nub x)(P \mid Q) :: z:A[\omega]$, (2)
 $\Omega ; \G ; \D \vdash P :: x{:}B[\omega'']$, and (3)
 $\Omega ; \G ; \D' , x{:}B[\omega'] \vdash Q :: z{:}A[\omega]$.
If $(\nub
x)(P \mid Q) \tra{} (\nub x)(P' \mid Q')$ then:
(a) $\Omega ; \G ; \D \vdash P' :: x'{:}B'[\omega'']$, for some
$x',B',\omega''$;
(b) $\Omega ; \G , \D' , x'{:}B'[\omega''] \vdash Q' :: z{:}A[\omega]$;
(c) $\omega \prec^* \omega''$.
\end{theorem}

\section{Domain-Aware Multiparty Session Types}
\label{sec:hmpst}

We now shift our attention to multiparty session
types~\cite{DBLP:conf/popl/HondaYC08}.  We consider the standard
ingredients: \emph{global types}, \emph{local types}, and the
\emph{projection function} that connects the two.  Our global types
include a new domain-aware construct,
$\gmoves{\pt{p}}{\ptset{q}}{\omega}{G_1}{G_2}$; our local types
exploit the hybrid session types from \defref{d:hillprops}.  Rather
than defining a separate type system based on local types for the
process model of \secref{sec:proc-model}, our analysis of multiparty
protocols extends the approach defined
in~\cite{DBLP:conf/forte/CairesP16}, which uses \emph{medium
  processes} to characterize correct multiparty implementations. The
advantages are twofold: on the one hand, medium processes provide a
precise semantics for global types; on the other hand, they enable the
principled transfer of the correctness properties established in
\secref{sec:hill} for binary sessions (type preservation, global
progress, termination, domain preservation) to the multiparty setting.
Below,
\emph{participants} 
 are ranged over by $\pt{p},  \pt{q}, \pt{r},   \ldots$;  
 we write $\ptset{\pt{q}}$ to denote a finite set of participants $\pt{q}_1, \ldots, \pt{q}_n$.

Besides the new domain-aware global type, 
our syntax of global types includes constructs from \cite{DBLP:conf/popl/HondaYC08,DBLP:conf/icalp/DenielouY13}.
We consider value  passing 
in branching (cf. $U$ below), fully supporting delegation. 
 To streamline the presentation,  we consider 
  global types without recursion.

\begin{definition}[Global and Local Types]\label{d:gltypes}
Define global types ($G$) and local types ($T$) as

\vspace{-0.4cm}
{\small\[
\begin{array}{rcl}
&U  ::=    & \mathsf{bool} \sep \mathsf{nat} \sep \mathsf{str} \sep\ldots \sep T
\\
&G  ::= &  \gend \sep \gto{p}{q}\{\lb{l}_i\langle U_i\rangle.G_i\}_{i \in I} 
          \sep {\gmoves{\pt{p}}{\ptset{q}}{\omega}{G_1}{G_2}}
\\
&T  ::= &  \lend \sep \mathtt{p}?\{\lb{l}_i\langle U_i\rangle.T_i\}_{i \in I} \sep \mathtt{p}!\{\lb{l}_i\langle U_i\rangle.T_i\}_{i \in I}  
          \sep {\forall \alpha.  T \sep \exists \alpha. T}
  \sep  {  @_\alpha T } \sep {\here{\alpha.T}}
\end{array}
\]}
\end{definition}

The completed global type is denoted $\gend$.
Given a finite $I$  and pairwise different labels,
$ \gto{p}{q}\{\lb{l}_i\langle U_i\rangle.G_i\}_{i \in I}$ specifies that by choosing label $\lb{l}_i$,
participant
 \pt{p} may send a message of type $U_i$ to 
  participant
 \pt{q}, and then continue as  $G_i$.
 We decree $\pt{p} \neq \pt{q}$, 
 so reflexive interactions are disallowed. 
  The global type \gmoves{\pt{p}}{\ptset{\pt{q}}}{\omega}{G_1}{G_2}
  specifies the  migration of participants $\pt{p}, \ptset{\pt{q}}$ to domain $\omega$ in order to perform the \emph{sub-protocol}~$G_1$; this migration is lead by $\pt{p}$. Subsequently, all of $\pt{p}, \ptset{\pt{q}}$ migrate  from $\omega$  back to their original domains and protocol $G_2$ is executed. 
  This intuition will be made precise by the medium processes for global types (cf.~\defref{d:ether}).
  Notice that $G_1$ and $G_2$ may involve different sets of  participants.
  In writing  $\gmoves{\pt{p}}{\ptset{\pt{q}}}{\omega}{G_1}{G_2}$ we assume two natural conditions: 
  (a)~all migrating participants intervene in the sub-protocol (i.e., the set of participants of $G_1$ is exactly $\pt{p}, \ptset{\pt{q}}$)
  and 
  (b)~domain $\omega$ is accessible (via $\prec$) by all these migrating participants in $G_1$.
  \btnote{While subprotocols and session delegation may appear as similar, delegation supports a
    different idiom altogether, and has no support for domain
    awareness. Unlike delegation, with subprotocols we can specify a point
    where some of the participants perform a certain
    protocol {\em within the same multiparty session} and then
    return to the main session as an ensemble.}

  \begin{definition}
  The set of participants of $G$ (denoted \partp{G}) is defined as:
~$\partp{\gend} 
=\emptyset$,
$\partp{\gto{p}{q}\{\lb{l}_i\langle U_i\rangle.G_i\}_{i \in I}} = \{\mathtt{p} , \mathtt{q}\} \cup \bigcup_{i \in I} \partp{G_i}$,
 {$\partp{\gmoves{\pt{p}}{\ptset{\pt{q}}}{\omega}{G_1}{G_2}} = \{\pt{p}\} \cup \ptset{\pt{q}} \cup \partp{G_1} \cup \partp{G_2}$}.
We sometimes write $\pt{p} \in G$ to mean $\pt{p} \in \partp{G}$.
\end{definition}

Global types are projected onto participants so as to obtain local types. 
  The terminated local type is  $\lend$.
    The local type $ \mathtt{p}?\{\lb{l}_i\langle U_i\rangle.T_i\}_{i \in I}$
  denotes an offer of a  set of labeled alternatives; 
  the 
  local type  $\mathtt{p}!\{\lb{l}_i\langle U_i\rangle.T_i\}_{i \in I}$ denotes a behavior that chooses one of such alternatives.
 {Exploiting the domain-aware framework in \secref{sec:hill}, we introduce four new local types. 
They increase the expressiveness of standard local types by specifying
universal and existential quantification over domains ($\forall \alpha. T$ and $\exists \alpha. T$),  migration to a specific domain ($@_\alpha T$), and a  reference to the current domain ($\here{\alpha.\,T}$, with $\alpha$ occurring in $T$).}

We now define \emph{(merge-based) projection} for global types~\cite{DBLP:conf/icalp/DenielouY13}.
To this end, we rely on a \emph{merge} operator on local types, which
in our case considers messages $U$.

\begin{definition}[Merge]\label{d:mymerg}
We define $\fuse$ 
as the commutative partial operator 
on base and local types such that
$\mathsf{bool} \fuse \mathsf{bool} =  \mathsf{bool}$ (and analogously for other base types), and 
\begin{enumerate}
\item $T \fuse T = T$, where $T$ is one of the following: $\lend$,  $\pt{p}!\{\lb{l}_i \langle U_i\rangle.T_i\}_{i \in I}$, $@_\omega T$, $\forall \alpha. T$, or $\exists\alpha. T$; 


\item $\pt{p}?\{\lb{l}_k\langle U_k\rangle.T_k\}_{k \in K} \fuse \mathtt{p}?\{\lb{l}'_j\langle U'_j\rangle.T'_j\}_{j \in J} =$ \\ 
 $~\qquad \pt{p}?\big(\{\lb{l}_k\langle U_k\rangle.T_k\}_{k \in K\setminus J} \cup 
        \{\lb{l}'_j\langle U'_j\rangle.T'_j\}_{j \in J\setminus K}  \cup \{\lb{l}_l\langle U_l \fuse U'_l \rangle.(T_l \fuse T'_l)\}_{l \in K\cap J}\big)
        $ 
\end{enumerate}
and is undefined otherwise.
\end{definition}

\noindent Therefore, 
for $U_1 \fuse U_2$ to be defined there are two options:
(a) $U_1$ and $U_2$ are identical base, terminated, selection, or ``hybrid'' local types;
 (b) $U_1$ and $U_2$ are branching types, but not necessarily identical:
they may offer different options but with the condition that the
behavior in labels occurring in both $U_1$ and $U_2$ must be mergeable.

To define projection and medium processes for the new global type $\gmoves{\pt{p}}{\ptset{\pt{q}}}{\omega}{G_1}{G_2}$, we require 
ways of ``fusing'' local types and processes. The intent is to  capture in a single (sequential) specification the behavior of two distinct (sequential) specifications, i.e., those corresponding to protocols $G_1$ and $G_2$. For local types, we have 
the following definition, \btnote{which safely appends a local type to another:}
\begin{definition}[Local Type Fusion]
\label{d:ltfuse}
The \emph{fusion} of $T_1$ and $T_2$, written $T_1 \ltfuse T_2$, is
given by:

\vspace{-0.5cm}
{\small\[
  \begin{array}{lcllcl}
\mathtt{p}!\{\lb{l}_i\langle U_i\rangle.T_i\}_{i \in I}  \ltfuse T  & =&
\mathtt{p}!\{\lb{l}_i\langle U_i\rangle.(T_i \ltfuse T) \}_{i \in I} 
    &
      \lend \ltfuse T & = & T\\
\mathtt{p}?\{\lb{l}_i\langle U_i\rangle.T_i\}_{i \in I} \ltfuse T & = &
\mathtt{p}?\{\lb{l}_i\langle U_i\rangle.(T_i \ltfuse T) \}_{i \in I} 
&
(\exists \alpha.   T_1) \ltfuse T & = &
\exists \alpha.   (T_1 \ltfuse T)
\\
(\forall \alpha.   T_1) \ltfuse T & = &
\forall \alpha.   (T_1 \ltfuse T)
&
  (@_\alpha T_1) \ltfuse T & = &@_\alpha (T_1 \ltfuse T)\\
    (\here{\alpha.T_1}) \ltfuse T & = &\here{\alpha.(T_1\ltfuse T)}                                   
  \end{array}
  \]}
\end{definition}
This way, e.g., if $T_1 = \exists \alpha.@_\alpha\,\mathtt{p}?\{\lb{l}_1\langle \mathsf{Int} \rangle.\lend \, , \lb{l}_2\langle \mathsf{Bool} \rangle.\lend\}$
and $T_2 = @_\omega \, \mathtt{q}!\{\lb{l}\langle \mathsf{Str}\rangle.\lend\}$, then
$T_1 \ltfuse T_2 = \exists \alpha.@_\alpha\,\mathtt{p}?\{\lb{l}_1\langle \mathsf{Int} \rangle.@_\omega \, \mathtt{q}!\{\lb{l}\langle \mathsf{Str}\rangle.\lend\} \, , \lb{l}_2\langle \mathsf{Bool} \rangle.@_\omega \, \mathtt{q}!\{\lb{l}\langle \mathsf{Str}\rangle.\lend\}\}$.
We can now define:


\begin{definition}[Merge-based Projection~\cite{DBLP:conf/icalp/DenielouY13}]\label{d:proj}
Let $G$ be a global type.
The \emph{merge-based projection} of $G$ under participant $\pt{r}$, 
denoted \proj{G}{\pt{r}}, is defined as 
$\proj{\gend}{\pt{r}}  =  \lend$ and 
\begin{itemize}

\item $\proj{\gto{p}{q}\{\lb{l}_i\langle U_i\rangle.G_i\}_{i \in I}}{\pt{r}}  = 
{\small\begin{cases}
\pt{p}!\{\lb{l}_i \langle U_i\rangle.\proj{G_i}{\pt{r}}\}_{i \in I} & \text{if $\pt{r} = \pt{p}$} \\
\pt{p}?\{\lb{l}_i \langle U_i\rangle.\proj{G_i}{\pt{r}}\}_{i \in I} & \text{if $\pt{r} = \pt{q}$} \\
\fuse_{i \in I} \,\proj{G_i}{\pt{r}} & \text{otherwise ($\fuse$ as in \defref{d:mymerg})} 
\end{cases}} $

											
\item {$\proj{(\gmoves{\pt{p}}{\ptset{\pt{q}}}{\omega}{G_1}{G_2})}{\pt{r}} =
{\small\begin{cases}
\here{\beta.(\exists \alpha. @_\alpha \,\proj{G_1}{\pt{r}}) \ltfuse @_\beta \, \proj{G_2}{\pt{r}}}  & \text{if $\pt{r} = \pt{p}$}
\\
\here{\beta.(\forall \alpha. @_\alpha \, \proj{G_1}{\pt{r}}) \ltfuse @_\beta \, \proj{G_2}{\pt{r}}}  & \text{if $\pt{r} \in \ptset{\pt{q}}$}
\\
\proj{G_2}{\pt{r}}  & \text{otherwise}
\end{cases}} $}

\end{itemize}
When \btnote{no side condition} holds, the map is undefined.
\end{definition}

 {\noindent The  projection for the   type $\gmoves{\pt{p}}{\ptset{\pt{q}}}{w}{G_1}{G_2}$ is one of the key points in our analysis. 
The local type for $\pt{p}$, the leader of the migration, starts by binding the identity of its current domain (say, $\omega_\pt{p}$) to $\beta$. 
Then, the (fresh) domain  $\omega$ is communicated, and there is a migration step to $\omega$, which is where protocol $\proj{G_1}{\pt{p}}$ will be performed. Finally, there is a migration step from $\omega$ back to $\omega_\pt{p}$; once there, the protocol $\proj{G_2}{\pt{p}}$ will be performed.
The local type for all of $\pt{q}_i \in \ptset{\pt{q}}$ follows accordingly: they expect $\omega$ from $\pt{p}$; the migration from their original domains to $\omega$ (and back) is as for $\pt{p}$. For participants in  $G_1$, the fusion on local types (\defref{d:ltfuse}) defines a local type that includes the actions for $G_1$ but also for $G_2$, if any: a participant  in $G_1$ need not be involved in $G_2$.  
Interestingly, the resulting local types 
$\here{\beta.(\exists \alpha. @_\alpha \,\proj{G_1}{\pt{p}}) \ltfuse @_\beta \, \proj{G_2}{\pt{p}}}$
and
$\here{\beta.(\forall \alpha. @_\alpha \,\proj{G_1}{{\pt{q}_i}}) \ltfuse @_\beta \, \proj{G_2}{{\pt{q}_i}}}$
define a precise combination of hybrid connectives whereby each migration step is bound  by a quantifier or the current domain.
}

The following notion of \emph{well-formedness} for global types is standard:

\begin{definition}[Well-Formed Global Types~\cite{DBLP:conf/popl/HondaYC08}]\label{d:wfltypes}
We say that global type $G$ is \emph{well-formed (WF, in the following)} if 
the   projection $\proj{G}{\pt{r}}$ is defined for all $\pt{r} \in G$.
\end{definition}

\subparagraph{Analyzing Global Types via Medium Processes}
A \emph{medium process} is a well-typed process from \secref{sec:proc-model} that captures the communication behavior of the domain-aware global types of \defref{d:gltypes}. Here we define medium processes and establish two fundamental characterization results for them (Theorems \ref{l:ltypesmedp} and \ref{l:medltypes}).
We shall consider names \emph{indexed by participants}: given a name $c$ and a participant $\pt{p}$, 
we use $c_{\pt{p}}$ to denote the name along which the session behavior of $\pt{p}$ will be made available. 
This way, if  $\pt{p} \neq \pt{q}$ then $c_{\pt{p}} \neq c_{\pt{q}}$.
%
%
  %
To define mediums, we need to \btnote{append or fuse} sequential processes, just as \defref{d:ltfuse} fuses local types:

\begin{definition}[Fusion of Processes]
\label{d:prfuse}
We define $\prfuse$ as the partial operator on well-typed processes
such that (with $\pi \in \{c(y), c\langle \omega \rangle, c(\alpha),
c\langle y @ \omega\rangle, c(y @ \omega), \mysel{c}{\lb{l}}\}$) :

\vspace{-0.5cm}
{\small\[
  \begin{array}{lclrcl}

c\out{y}.(\linkr{u}{y} \mid P) \prfuse Q & \triangleq &
                                                        c\out{y}.(\linkr{u}{y} \mid (P\prfuse Q))
    &           \mathbf{0} \prfuse Q & \triangleq &Q\\                                        

\mycasebig{c}{\lb{l}_i : P_i}{i \in I} \prfuse Q & \triangleq &
                                                                \mycasebig{c}{\lb{l}_i
                                                                : (P_i
                                                                \prfuse
                                                                Q) }{i
                                                                \in I}

&
(\pi.P) \prfuse Q & \triangleq & \pi.(P \prfuse Q)                                           
  \end{array}
  \]}
and is undefined otherwise.
\end{definition}
The previous definition suffices to define a medium process (or simply \emph{medium}), which 
uses indexed names to uniformly capture the behavior of a global type:

\begin{definition}[Medium Process]\label{d:ether}
Let $G$ be a global type (cf. \defref{d:gltypes}),
$\tilde{c}$ be a set of indexed names, and $\tilde{\omega}$ a set of domains.
The \emph{medium} of $G$, denoted \etherp{G}{\tilde{c}}{\tilde{\omega}}, is  defined as:

\vspace{-0.3cm}
{\small\[
\begin{cases}
\mathbf{0} & 
\text{if $G = \gend$}
\vspace{1mm}
\\
\mycasebig{c_\pt{p}}{\lb{l}_i :
  c_\pt{p}(u).\mysel{c_\pt{q}}{\lb{l}_i};\outp{c_\pt{q}}{v}.(
  \linkr{u}{v} \para 
\etherp{G_i}{\tilde{c}}{\tilde{\omega}} )}{i \in I}
& 
\text{if $G = \gto{p}{q}\{\lb{l}_i\langle U_i\rangle.G_i\}_{i \in I}$}
\vspace{1mm} \\
c_\pt{p}(\alpha).c_{\pt{q}_1}\out{\alpha}. \cdots .c_{\pt{q}_n}\out{\alpha}. 
& 
\text{if $G = \gmoves{\pt{p}}{\pt{q}_1, \ldots, \pt{q}_n}{w}{G_1}{G_2}$}
\\
\quad
c_\pt{p}(y_\pt{p}@\alpha).c_{\pt{q}_1}(y_{\pt{q}_1}@\alpha).\cdots.c_{\pt{q}_n}(y_{\pt{q}_n}@\alpha).
&
\\
\quad \quad \etherp{G_1}{\tilde{y}}{\tilde{\omega}\{\alpha/\omega_{\pt{p}},
  \dots , \alpha/\omega_{\pt{q}_n}\} } \,\,\prfuse 
  \\
 \quad \qquad
 (y_\pt{p}(m_\pt{p}@\omega_\pt{p}). y_{\pt{q}_1}(m_{\pt{q}_1}@\omega_{\pt{q}_1}).
\cdots. y_{\pt{q}_n}(m_{\pt{q}_n}@\omega_{\pt{q}_n}).
\\
\qquad \quad\quad\etherp{G_2}{\tilde{m}}{\tilde{\omega}}) & 
\end{cases} 
\]}
where $\etherp{G_1}{\tilde{c}}{\tilde{\omega}} \prfuse \etherp{G_2}{\tilde{c}}{\tilde{\omega}}$ 
is as in \defref{d:prfuse}.
\end{definition}

The medium  for
$G = \gto{p}{q}\{\lb{l}_i\langle U_i\rangle.G_i\}_{i \in I}$
exploits four prefixes to mediate in the interaction between the implementations of $\pt{p}$ and $\pt{q}$:
the first two prefixes (on name $c_\pt{p}$) capture the label selected
by $\pt{p}$ and the subsequently received value; the third and fourth
prefixes (on name $c_\pt{q}$) propagate the choice and forward the
value sent by $\pt{p}$ to $\pt{q}$. \btnote{We omit the forwarding and
  value exchange when the interaction does not involve a value payload.}
  
The medium for 
$G = \gmoves{\pt{p}}{\pt{q}_1, \ldots, \pt{q}_n}{w}{G_1}{G_2}$
 showcases the expressivity and convenience of our domain-aware
process framework. In this case, the medium's behavior takes place
through the following steps: First, $\etherp{G}{\tilde{c}}{\tilde{\omega}}$ inputs a domain identifier
(say, $\omega$) from  $\pt{p}$ which is forwarded to $\pt{q}_1,
\ldots, \pt{q}_n$, the other participants of $G_1$. Secondly, the
roles $\pt{p}, \pt{q}_1, \ldots, \pt{q}_n$ migrate from their domains 
$\omega_{\pt{p}}, \omega_{\pt{q}_1}\ldots, \omega_{\pt{q}_n}$ to
$\omega$.
At this point, the medium for $G_1$ can execute, keeping track the
current domain $\omega$ for all participants. Finally, the
participants of $G_1$ migrate back to their original domains and the
medium for $G_2$ executes.

Recalling the domain-aware global type of 
\secref{sect:intro}, we produce its medium process:

\vspace{-0.1cm}
{\small\begin{tabbing}
     $\mycasebig{c_\pt{cl}}{\lb{request} : c_\pt{cl}(r).
     \mysel{c_\pt{mw}}{\lb{request}};
     \outp{c_\pt{mw}}{v}.(\linkr{r}{v} \mid$
     \\
     \quad $\mycasebig{c_\pt{mw}}{$\=$\lb{reply} :
       c_\pt{mw}(a).\mysel{c_\pt{cl}}{\lb{reply}};
       \outp{c_\pt{cl}}{n}.(\linkr{a}{n} \mid
       \mycasebig{c_\pt{mw}}{\lb{done} : \mysel{c_\pt{serv}}{\lb{done}};\zero }{} ),$
       \\
       \>$ \lb{wait} : \btnote{\mysel{c_\pt{cl}}{\lb{wait}};}\mycasebig{c_\pt{mw}}{$\=$\lb{init} :
         \mysel{c_\pt{serv}}{\lb{init}};
         c_\pt{mw}(w_\pt{priv}).c_\pt{serv}\langle
         w_\pt{priv}\rangle.$  \\
         \>\>$c_\pt{mw}(y_\pt{mw}@ w_\pt{priv}).
         c_\pt{serv}(y_\pt{serv}@ w_\pt{priv}).
         \etherp{\mathit{Offload}}{y_\pt{mw},y_\pt{serv}}{\tilde{w_\pt{priv}}}
         \,\,\prfuse$\\
         \>\>$(y_\pt{mw}(z_\pt{mw}@w_\pt{mw}).
         y_\pt{serv}(z_\pt{serv}@w_\pt{serv}).$\\
         \>\>$\mycasebig{z_\pt{mw}}{\lb{reply} : z_\pt{mw}(a).
         \mysel{c_\pt{cl}}{\lb{reply}};\outp{c_\pt{cl}}{n}.(\linkr{a}{n}
       \mid \zero)}{} )
       }{}  }{}  ) }{}$
 \end{tabbing}}
The medium ensures the client's domain remains fixed through the entire
 interaction, regardless of whether the middleware chooses to
 interact with the server. 
 This showcases
 how our medium transparently manages domain migration of participants.


\vspace{-0.4cm}
\subparagraph{Characterization Results}
We state results that offer a sound and complete account of the relationship between: 
(i)~a global type $G$ (and its  local types), 
(ii)~its  medium process $\etherp{G}{\tilde{c}}{\tilde{\omega}}$, and 
(iii)~process implementations for the participants $\{\pt{p}_1, \ldots, \pt{p}_n\}$ of $G$.
In a nutshell, 
these results say that 
the typeful composition of $\etherp{G}{\tilde{c}}{\tilde{\omega}}$ with 
processes for each $\pt{p}_1, \ldots, \pt{p}_n$ (well-typed in the system of \secref{sec:hill}) performs the intended global type. 
Crucially, these  processes reside in distinct domains and can be independently developed, guided by their local type---they need not know about the medium's existence or structure.
The results generalize those in~\cite{DBLP:conf/forte/CairesP16} to the domain-aware setting.
Given a  global type $G$ with $\partp{G} = \{\pt{p}_1, \ldots, \pt{p}_n\}$, 
 below we write $\npart{G}$ to denote the set of indexed names
 $\{c_{\pt{p}_1}, \ldots, c_{\pt{p}_n}\}$. 
We  define: 

\begin{definition}[Compositional Typing]\label{d:compty}
We say
$\Omega; \G; \D \vdash \etherp{G}{\tilde{c}}{\tilde{\omega}} ::z{:}C$
is a \emph{compositional typing} 
if:  (i)~it is a valid typing derivation;
(ii)~$ \npart{G} \subseteq \dom{\D}$; 
and
 (iii)~$C = \one$. 
\end{definition}

\noindent A compositional typing 
says that  $\etherp{G}{\tilde{c}}{\tilde{\omega}}$ depends on
behaviors associated to each participant of $G$; 
it also specifies that $\etherp{G}{\tilde{c}}{\tilde{\omega}}$
does not offer any behaviors of its own.

The following definition relates 
binary session types and local types: 
the main difference is that
the former do not mention participants.
Below,   $B$ ranges over base types ($\mathsf{bool}, \mathsf{nat}, \ldots$). 
\begin{definition}[Local Types$\to$Binary Types]\label{d:loclogt}
Mapping \lt{\cdot} from local types $T$ (\defref{d:gltypes}) 
into binary types $A$ (\defref{d:hillprops}) is 
inductively defined as $\lt{\lend} = \lt{B} =  \one$ and

\vspace{-0.4cm}
{\small
  \[
\begin{array}{lcllcl}
\lt{\mathtt{p}!\{\lb{l}_i\langle U_i\rangle.T_i\}_{i \in I}}  & = &  \myoplus{\lb{l}_i : \lt{U_i} \otimes \lt{T_i}}{i \in I} 
  &
    \lt{\forall \alpha.  T} & =  &\forall \alpha.  \lt{T}
\\
 \lt{\mathtt{p}?\{\lb{l}_i\langle U_i\rangle.T_i\}_{i \in I}} & =  &  \mywith{\lb{l}_i : \lt{U_i} \lolli \, \lt{T_i}}{i \in I} 
&

\lt{\exists \alpha.  T} & = & \exists \alpha.   \lt{T}
\\
  \lt{  @_\omega T} & = &  @_\omega \lt{T}
  &
  \lt{ \here{\alpha.T}} & = &\,\here{\alpha.\lt{T}}
\end{array}
\]
}
\end{definition}

Our first characterization result ensures that well-formedness of a global type $G$  guarantees the typability of its medium  $\etherp{G}{\tilde{c}}{\tilde{\omega}}$ using binary session types. Hence, it ensures that multiparty protocols can be analyzed by composing the medium with independently obtained, well-typed implementations for each protocol participant. Crucially, the resulting well-typed process will inherit all correctness properties ensured by binary typability established in \secref{sec:hill}.

\begin{restatable}[Global Types $\to$ Typed Mediums]{theorem}{thmtypesmedp}\label{l:ltypesmedp}
If  $G$ is 
WF
with $\partp{G}\!= \{\pt{p}_1, \ldots, \pt{p}_n\}$
then 
$
\Omega; \G; c_{\pt{p}_1}{:}\lt{\proj{G}{\pt{p}_1}}[\omega_1], \ldots, c_{\pt{p}_n}{:}\lt{\proj{G}{\pt{p}_n}}[\omega_n] \vdash  \etherp{G}{\tilde{c}}{\tilde{\omega}}::z:\one[\omega_m]
$
 is a compositional typing,  
for some $\Omega$, $\G$, with $\tilde{\omega} = \omega_1, \ldots, \omega_n$.
We assume that $\omega_i \prec \omega_m$
 for all $i \in \{1, \ldots, n\}$ (the medium's domain is accessible by all), and that $i\neq j$ implies $\omega_i \neq \omega_j$.
\end{restatable}

\noindent 
The second  characterization result, given next,  
is the converse of Theorem~\ref{l:ltypesmedp}:
  binary typability precisely delineates the interactions that underlie well-formed multiparty protocols.
  We need an auxiliary relation on local types, written $ \subt$,
that relates types with branching and ``here'' type operators, which
have silent process interpretations (cf. Figure~\ref{fig:phillrules}
and \cite{longversion}). 
First, 
we have $T_1  \subt T_2$ if there is a $T'$ such that $T_1 \fuse T' = T_2$ (cf. \defref{d:mymerg}). 
Second, 
we have $T_1  \subt T_2$ if 
(i)~$T_1 = T'$ and $T_2 =\, \here{\alpha.T'}$ and $\alpha$ does not occur in $T'$;
but also if (ii)~$T_1 =\, \here{\alpha.T'}$ and $T_2 = T'\subst{\omega}{\alpha}$.
(See \cite{longversion} for a formal definition of  $\subt$).

\begin{restatable}[Well-Typed Mediums $\to$ Global Types]{theorem}{thmmedltypes}\label{l:medltypes}
Let $G$ be a global type (cf. \defref{d:gltypes}). \\ 
If 
$\Omega; \G; c_{\pt{p}_1}{:}A_1[\omega_1], \ldots, c_ {\pt{p}_n}{:}A_ n[\omega_n] \vdash  \etherp{G}{\tilde{c}}{\tilde{\omega}}::z:\one[\omega_m]$
is a compositional typing 
then $\exists 
T_1, \ldots, T_ n$ 
such that 
$\proj{G}{\mathtt{p}_j} \subt T_j$ and 
$\lt{T_j} = A_j$, 
for all $\pt{p}_j \in \partp{G}$.
\end{restatable}


\noindent
The above theorems offer a \emph{static guarantee} that connects multiparty protocols and well-typed processes. They can be used to establish also \emph{dynamic guarantees} relating the behavior of a global type $G$ and that of its associated set of \emph{multiparty systems} (i.e., the typeful composition of $\etherp{G}{\tilde{c}}{\tilde{\omega}}$ with processes for each of $\pt{p}_i \in \partp{G}$).
These dynamic guarantees 
can be easily obtained by combining Theorems \ref{l:ltypesmedp} and \ref{l:medltypes} with the approach  in~\cite{DBLP:conf/forte/CairesP16}.




\section{Related Work}\label{s:relw}

\jpnote{There is a rich history of works on the logical foundations of concurrency (see, e.g.,~\cite{DBLP:journals/tcs/BellinS94,DBLP:conf/tapsoft/GirardL87,DBLP:journals/tcs/Abramsky93,DBLP:journals/entcs/Beffara06}), which has been extended to  
session-based concurrency by
Wadler~\cite{DBLP:journals/jfp/Wadler14}, 
Dal Lago and Di Giamberardino~\cite{DBLP:journals/corr/abs-1108-4467},
and others.}
Medium-based
analyses of multiparty sessions were developed
in~\cite{DBLP:conf/forte/CairesP16} and used in an account of multiparty sessions in an extended classical linear
logic~\cite{DBLP:conf/concur/CarboneLMSW16}.

Two salient calculi with distributed features are the Ambient
calculus~\cite{DBLP:journals/tcs/CardelliG00}, in which processes  
move across \emph{ambients} (abstractions of administrative domains),
and the \emph{distributed
  $\pi$-calculus}~(\textsc{Dpi})~\cite{DBLP:journals/iandc/HennessyR02},
which extends the $\pi$-calculus with flat locations, 
local communication, and process migration.
%
 %
While domains in our model may be read as locations, this is just one specific interpretation; 
they admit various alternative readings
(e.g.,~administrative domains, security-related levels), leveraging the
partial view of the domain hierarchy.
Type systems for Ambient calculi such as~\cite{DBLP:journals/iandc/CardelliGG02,DBLP:journals/cl/BugliesiC02}
enforce security and
communication-oriented properties in terms of ambient movement but 
do not cover issues of structured interaction, central in our work. 
Garralda et
al.~\cite{DBLP:conf/ppdp/GarraldaCD06} integrate binary sessions in an
Ambient calculus, ensuring that session protocols are
undisturbed by ambient mobility.  In contrast, our type system ensures that both migration and
communication are safe and, for the first time in such a setting,
satisfy global progress (i.e.,~session protocols never jeopardize
migration and vice-versa).

The multiparty sessions with nested protocols of Demangeon and Honda~\cite{DBLP:conf/concur/DemangeonH12} include 
a nesting construct that is
similar to our \jpnote{new global type $\gmoves{\pt{p}}{\ptset{\pt{q}}}{w}{G_1}{G_2}$}, which also introduces
nesting.  The focus in~\cite{DBLP:conf/concur/DemangeonH12} is on modularity in choreographic
programming; domains nor domain migration are not addressed.
The nested protocols in~\cite{DBLP:conf/concur/DemangeonH12} can have {\em local} participants and
may be parameterized on data from previous actions.  We conjecture
that our approach can accommodate local participants in a
similar way. Data parameterization can be transposed to our logical
setting via dependent session types
\cite{Toninho2011,DBLP:conf/fossacs/ToninhoY18}.
Asynchrony and recursive behaviors can also
be integrated by exploiting existing logical foundations \cite{DBLP:conf/csl/DeYoungCPT12,DBLP:conf/tgc/ToninhoCP14}. 

Balzer et al.~\cite{DBLP:conf/esop/BalzerTP19} overlay a notion of
world and accessibility on a system of {\em shared} session types to
ensure deadlock-freedom.  Their work differs substantially from ours:
\btnote{they instantiate accessibility as a partial-order, equip
sessions with multiple worlds and are not conservative wrt
linear logic,} 
being closer to partial-order-based typings for
deadlock-freedom~\cite{DBLP:conf/concur/Kobayashi06,DBLP:conf/csl/Padovani14}.
\section{Concluding Remarks}\label{sect:concl}

We developed a Curry-Howard interpretation of hybrid linear logic
as domain-aware session types.  
Present in processes and
types, 
domain-awareness   \jpnote{can account} for 
scenarios where domain
information \btnote{is} only determined at runtime. 
The resulting type system features 
strong correctness properties 
 for well-typed 
processes (session fidelity, global progress, termination). Moreover, 
by leveraging a {\em  parametric} accessibility relation,  
it rules out processes
that communicate with inaccessible domains, thus going beyond the scope of previous works.

As an application of our 
framework, we presented the first systematic study of domain-awareness in a
{\em multiparty} setting,  considering
multiparty sessions with domain-aware migration and communication whose semantics is given by 
a typed (binary) medium
process that orchestrates the multiparty protocol. 
Embedded in a fully distributed domain structure, our medium
is shown to strongly encode domain-aware multiparty sessions; it 
naturally allows us to transpose the correctness properties of our logical
development to the multiparty setting.

Our work opens up interesting avenues for future work. 
Mediums can be seen as \emph{monitors} that enforce the
 specification of a domain-aware
multiparty session. 
We plan to investigate contract-enforcing mediums
building upon works such as~\cite{DBLP:conf/esop/GommerstadtJP18,DBLP:conf/popl/JiaGP16,DBLP:journals/fmsd/DemangeonHHNY15}, which
study runtime monitoring in session-based
systems. 
Our enforcement of communication across accessible domains suggests
high-level similarities with
information flow
analyses in 
 multiparty sessions  (cf.~\cite{DBLP:conf/concur/CapecchiCDR10,DBLP:journals/corr/abs-1108-4465,DBLP:journals/fac/CastellaniDP16}), \btnote{but does not capture the directionality needed to
  model such analyses outright}.
It would be insightful to establish the precise relationship with such prior works.
\bibliography{referen}

\appendix

\section{Appendix}\label{app:app}

\subsection{Structural Congruence}\label{app:procs}
\begin{definition}
  \label{def:struct-cong} {\em Structural congruence} ($P \equiv Q$)
  is the least congruence relation on processes such that
  \[
\begin{array}{c}
  P \para \zero  \equiv P  \qquad   P \equiv_{\alpha} Q \Rightarrow P \equiv Q \qquad  (\nub x_{})\zero \equiv \zero  \qquad \linkr{x}{y} \equiv \linkr{y}{x}  \qquad
  P \para Q \equiv Q \para P\\ P \para (Q \para R) \equiv (P \para Q) \para R  \qquad
    x \not\in\fn{P} \Rightarrow P \para (\nub x_{})Q \equiv (\nub x_{})(P \para Q)   \\
       \qquad  (\nub x_{})(\nub y_{})P \equiv (\nub y_{})(\nub x_{})P  
\end{array}
\]
\end{definition}

\subsection{Labeled Transition System}\label{app:lts}
Some technical results rely on labeled transitions rather than on reduction.
To characterize the interactions of a well-typed process with its environment,
we extend
the 
early labeled transition system (LTS) for the
$\pi$-calculus~\cite{sangiorgi-walker:book} with 
labels and transition rules for  choice, migration, and forwarding constructs.  
A transition
$P\tra{\,\labelset\,}Q$ denotes that 
$P$ may evolve to 
$Q$
by performing the action represented by label $\labelset$. 
Transition labels are defined below:
\begin{eqnarray*}
\labelset & ::= &  \tau  \sep x(y) \sep x(\wtag) \sep x.\lb{l} \sep  x.\lmigrate  
           \sep \ov{x \, y } \sep  \ov{x\out{y_{}}} \sep \ov{x \, \wtag } \sep\ov{x.\lb{l}} \sep  \ov{x.\lmigrate}   
\end{eqnarray*}
Actions are name input $x(y)$, 
domain input $x(\wtag)$, 
the 
 offers
$x.\inl$ and $ x.\inr$, migration $x.\lmigrate$ and their matching
co-actions,
respectively the output $\ov{x \, y}$ and bound output $\ov{x\out y}$ actions, 
the domain output $\ov{x \, \wtag }$,
label selections
$\ov{x.\lb{l}}$ and $\ov{x.\lb{l}}$, and domain migration $\ov{x.\lmigrate}$. 
Both the bound output $\ov{x\out y}$ and migration action $\ov{x.\lmigrate}$ denote extrusion of a fresh name $y$
along 
$x$. Internal action is denoted by $\tau$. 
In general, an action
requires a matching 
co-action
in the environment to enable progress. 
%


\begin{definition}[Labeled Transition System]\label{def:lts} 
  The relation \emph{labeled transition} ($ P\tra{\labelset}Q$) is defined
  by the rules in Fig.~\ref{fig:LTS},
   subject to the side
  conditions: in rule $(\mathsf{res})$, we require $y_{}\not\in\fn{\labelset}$; 
  in rule $(\mathsf{par})$, we require $\bn{\labelset} \cap \fn{R} = \emptyset$; in rule
  $(\mathsf{close})$, we require $y_{}\not\in\fn{Q}$. We omit the symmetric versions
  of rules $(\mathsf{par})$, $(\mathsf{com})$, and $(\mathsf{close})$.
  
We write $subj(\lambda)$ for the subject of the action $\lambda$, that
is, the channel along which the action takes place. Weak transitions are defined as usual.
Let us write $\rho_1 \rho_2$ for the composition of relations $\rho_1, \rho_2$
and $\wtra{}$ for the reflexive, transitive closure of
$\tra{\tau}$. 
Notation $\wtra{\labelset}$ stands for $\wtra{~}\tra{\labelset}\wtra{~}$ (given $\labelset \neq \tau$)
and $\wtra{\tau}$ stands for $\wtra{}$.
We recall basic facts about reduction, structural congruence,
and labeled transition: closure of labeled transitions under
structural congruence, and coincidence of $\tau$-labeled transition
and reduction~\cite{sangiorgi-walker:book}: (1) if $P
\equiv\tra{\labelset}Q$ then $P \tra{\labelset}\equiv Q$; 
(2) $P\to Q$ iff $P \tra{\tau} \equiv Q$. 

\begin{figure}[t!]
{\small
 $$
\begin{array}{ccccc}
\inferrule*[Left=\name{$\mathsf{id}$}]{}{(\nub x_{})(\linkr{x_{}}{y_{}} \para P) \tra{\,\tau\,}  P\subst{y_{}}{x_{}}} 
\qquad
\inferrule[\name{$\mathsf{n.out}$}]{}{x_{} \out {y_{}}.P \tra{\ov{x \, y_{}}} P}
\hspace{0.6cm}
\inferrule[\name{$\mathsf{n.in}$}]{}{x_{}(y).P \tra{x_{}(z_{})} P \subst{z_{}}{y}}
\vspace{0.2cm}
\\
\inferrule[\name{$\mathsf{d.out}$}]{}{x_{} \out {\wtag}.P \tra{\ov{x \, \wtag}} P}
\hspace{0.6cm}
\inferrule[\name{$\mathsf{d.in}$}]{}{x_{}(\alpha).P \tra{x(\wtag)} P \subst{\wtag}{\alpha}}
\vspace{0.2cm}
\\
\inferrule*[left=\name{$\mathsf{\textmigrate}$}]{}{x_{}\out{y@\wtag}.P
  \tra{\ov{x.\lmigrate}} (\nub y) P}
\quad
\inferrule*[left=\name{$\mathsf{\textmigrate}'$}]{}{x_{}(z@\wtag).P \tra{x.\lmigrate} P\subst{y}{z}}
\vspace{0.2cm}
\\
\inferrule*[Left=\name{$\mathsf{par}$}]{P \tra{\labelset} Q}{P\para R \tra{\labelset} Q \para R}
\hspace{1.5cm} 
\inferrule*[Left=\name{$\mathsf{com}$}]{P\tra{\overline{\labelset}} P' \;\;\; Q \tra{\labelset} Q'} {P \para
  Q \tra{\tau} P' \para Q'}
\quad
\inferrule*[left=\name{$\mathsf{res}$}]{P \tra{\labelset} Q} {(\nub y)P \tra{\labelset} (\nub y_{})Q}
\quad
\inferrule*[left=\name{$\mathsf{open}$}]{P \tra{\ov{x \, y_{}}} Q} {(\nub y_{})P \tra{\overline{x\out{y}}} Q}
\\
\inferrule*[Left=\name{$\mathsf{close}$}]{P \tra{\overline{x\out{y}}} P' \;\;\; Q \tra{x(y)} Q'}
{P \para Q \tra{\tau} (\nub y_{})(P' \para Q')}
\qquad
\inferrule[\name{$\mathsf{rep}$}]{}{\bang x_{}(y_{}).P \tra{x(z)} P \subst{z_{}}{y_{}}\para \bang x_{}(y_{}).P}
\vspace{0.15cm} 
\\
\inferrule[\name{$\mathsf{l.out}$}]{}{\mysel{x}{\lb{l}_i};P \tra{ \overline{x.\lb{l}_i} } P}
\quad
\inferrule[\name{$\mathsf{l.in}$}]{}{\mycasebig{x}{\lb{l}_i : P_i}{i \in I} \tra{x.\lb{l}_i} P_i}
\end{array}
$$
}
\caption{\label{fig:LTS} Labeled Transition System.}
\end{figure}
\end{definition} 

\subsection{Omitted Typing Rules}
\label{app:trules}

{\small
  \[
\begin{array}{ccccc}
\inferrule*[left=\name{$\rgt\with$}]{\bluee{\Omega ; \G; \D} \vdash
  \redd{P_1 :: x{:}}\bluee{A_1 [\omega] \quad \dots \quad \Omega ; \G; \D} \vdash \redd{P_n :: x{:}}\bluee{A_n[\omega]}}{\bluee{\Omega ; \G; \D} \vdash \redd{\mycasebig{x}{\lb{l}_i : P_i}{i \in I}::  z{:}}\bluee{ \with\{\lb{l}_i :
                                               A_i\}_{i\in I}[\omega]}}
\\[1em]
\inferrule*[left=\name{$\lft\with_1$}]{\Gamma; \Delta, \redd{x{:}}A[\omega_2] \vdash  \redd{P:: z}{:}C[\omega_1]}
{\Gamma; \Delta, \redd{x{:}}\mywith{\lb{l}_i : A}{\{i\}}[\omega_2]  \vdash  \redd{\mysel{x}{\lb{l}_i};P:: z{:}}C[\omega_1]}
\\[1em]
\inferrule*[left=\name{$\lft\with_2$}]{\Gamma; \Delta, \redd{x{:}}\with\{\lb{l}_i{:}A_i\}_{i \in I}[\omega_2] \vdash  \redd{P:: z{:}}C[\omega_1] \quad k \not\in I}
{\Gamma; \Delta, \redd{x{:}}\mywith{\lb{l}_j{:}A_j}{j \in I\cup\{k\}}[\omega_2] \vdash  \redd{P:: z{:}}C[\omega_1]}
\\[1em]
\inferrule*[left=\name{$\rgt\oplus_1$}]{\Gamma; \Delta \vdash  \redd{P:: x{:}}A[\omega]}
{\Gamma; \Delta \vdash  \redd{\mysel{x}{\lb{l}_i};P::x{:}}\myoplus{\lb{l}_i : A}{\{i\}}[\omega]}
\qquad
\inferrule*[left=\name{$\rgt\oplus_2$}]{\Gamma; \Delta \vdash  \redd{P:: x{:}}\myoplus{\lb{l}_i : A_i}{i \in I}[\omega] \quad k \not\in I}
{\Gamma; \Delta \vdash  \redd{P:: x{:}}\myoplus{\lb{l}_j : A_j}{j \in I\cup\{k\}}[\omega]}
\\[1em]
  \inferrule*[left=\name{$\lft\oplus$}]{\begin{array}{c}
    \bluee{\Omega ;\G;  \D,\,} \redd{x{:}}\bluee{A_1[\omega_2]} \vdash
                                          \redd{Q_1 ::
                                          z{:}}\bluee{C[\omega_1]} 
                                          \quad \dots \quad
\bluee{\Omega ; \G; \D ,\,} \redd{x{:}}\bluee{A_n[\omega_2]} \vdash \redd{Q_n :: z{:}}\bluee{C[\omega_1]}\end{array}}{\bluee{\Omega ; \G; \D ,\,} \redd{x{:}}\bluee{\oplus\{\lb{l}_i :
                                               A_i\}_{i\in I}[\omega_2]} \vdash  \redd{\mycasebig{x}{\lb{l}_i : Q_i}{i \in I} :: z{:}}\bluee{C[\omega_1]}}
\vspace{0.15cm}
\\
\inferrule*[left=\name{$\lft\bang$}]{
\Omega ; \Gamma, \redd{u{:}} A[\omega_2] ; \Delta \vdash \redd{P:: z{:}}\bluee{C[\omega_1]}}
{\Omega ; \Gamma; \Delta, \redd{x{:}}\bang A[\omega_2] \vdash \redd{x(u).P:: z{:}}\bluee{C[\omega_1]}}
\quad
\inferrule*[left=\name{$\rgt\bang$}]{\Omega ; \Gamma; \cdot \vdash \redd{Q:: y{:}}A[\omega]}
{\Omega ; \Gamma; \cdot  \vdash \redd{\ov{x}\langle u \rangle.\bang
  u(y).Q:: x{:}}\bang A[\omega]}\;
  \\[1em]
\inferrule*[left=\name{$\cut^{\bang}$}]
{\bluee{\Omega ; \G ; \cdot} \vdash \redd{P :: x{:}}\bluee{A[\omega_1] \quad \Omega ; \G ,\,} \redd{u{:}}\bluee{A[\omega_1] ; \D}
  \vdash \redd{Q :: z{:}}\bluee{C[\omega_2]}}{\bluee{\Omega ; \G ; \D} \vdash \redd{(\nub u)(\bang u(x).P  \mid Q) :: z{:}}\bluee{C[\omega_2]}}
\end{array}
\]}

\subsection{Additional Lemmas for Type Preservation}
\label{app:tpres}

The development of type preservation extends that of
\cite{DBLP:conf/concur/CairesP10} to account for domain communication
and migration. The proof mainly relies on 
a series of reduction lemmas (one per session type connective that
produces observable process actions) that relate
process actions with parallel composition through the \name{$\cut$}
rule, which correspond to logical proof reductions.
For instance, the reduction lemma for $\tensor$ is:

\begin{lemma}[Reduction Lemma -- $\tensor$]\label{lemma:cp-otimes}
  Assume \\
  (a)~ $\Omega;\G ; \D_1 \vdash  P :: x {:} A_1\otimes A_2[\omega]$ with
  $P \stackrel{\ov{(\nu y)x\out y}}{\rightarrow} P'$; and\\
  ~(b) $\G ;
  \D_2, x{:} A_1\otimes A_2[\omega]\vdash  Q :: z{:} C[\omega']$ with $Q
  \stackrel{x(y)}{\rightarrow} \! Q'$.\\
Then:  $\Omega ; \G ; \D_1 , \D_2 \vdash (\nub x)(P' \mid Q') ::
  z{:}C[\omega']$
\end{lemma}

These lemmas carry over straightforwardly from
\cite{DBLP:conf/concur/CairesP10}. The new lemmas are: 
\begin{lemma}[Reduction Lemma - $\forall$]\label{lem:redforall}~ Assume \\
(a)~$\Omega ; \G ; \D_1 \vdash P :: x{:}\forall \alpha . A [\omega_2]$
  with $P \tra{x(w_3)} P'$
and \\ 
(b)~$\Omega ; \G ; \D_2 , x{:}\forall \alpha . A[\omega_2] \vdash Q ::
  z{:} C [\omega_1]$ with $Q \tra{\overline{x \,w_3}} Q'$. \\
Then: $\Omega ; \G ; \D_1 , \D_2 \vdash (\nub x)(P' \mid Q') ::
  z{:}C[\omega_1]$
\end{lemma}

\begin{lemma}[Reduction Lemma - $\exists$]\label{lem:redexists} ~ Assume \\
(a)~$\Omega ; \G ; \D_1 \vdash P :: x{:}\exists \alpha . A [\omega_2]$
  with $P \tra{\overline{x \,w_3}} P'$
and \\ 
(b)~$\Omega ; \G ; \D_2 , x{:}\exists \alpha . A[\omega_2] \vdash Q ::
  z{:} C [\omega_1]$ with $Q \tra{x(w_3)} Q'$. \\
Then: $\Omega ; \G ; \D_1 , \D_2 \vdash (\nub x)(P' \mid Q') ::
  z{:}C[\omega_1]$
\end{lemma}

\begin{lemma}[Reduction Lemma - $@$] ~ Assume \\
(a)~$\Omega ; \G ; \D_1 \vdash P :: x{:}@_{\omega} A [\omega']$
  with $P  \tra{\ov{x.\lmigrate}}  P'$
and \\ 
(b)~$\Omega ; \G ; \D_2 , x{:}@_{\omega} A[\omega'] \vdash Q ::
  z{:} C [\omega'']$ with $Q \tra{x.\lmigrate} Q'$. \\
Then: $\Omega ; \G ; \D_1 , \D_2 \vdash (\nub x)(P' \mid Q') ::
  z{:}C[\omega'']$
\end{lemma}

The proofs of the lemmas above follow by simultaneous induction on the two
given typing derivations, with Lemmas~\ref{lem:redforall} and~\ref{lem:redexists}
making use of Lemma~\ref{lem:wsubst}.
This development is essentially that of
\cite{DBLP:conf/esop/CairesPPT13} and \cite{Toninho2011} which
consider an extension of the core propositional system of~\cite{DBLP:conf/concur/CairesP10} with communication of types
(i.e.~polymorphism) and communication of data (i.e.~value
dependencies). 
By appealing to such lemmas, we can establish
type preservation for our system.

\subsection{Pre-congruence on Local Types}\label{app:pctypes}

The following definition is used in the proof of \thmref{l:medltypes}:

\begin{definition}
  \label{def:pctypes} We define $\subt$
  as the least pre-congruence relation on local types such that
  \[
\begin{array}{c}
  T_1 \subt T_1 \fuse T_2  
  \qquad   
  T_1 \subt T_2 \Rightarrow ~\here{\alpha.T_1} \subt T_2\subst{\omega}{\alpha} 
  \qquad  
  T \subt ~\here{\alpha.T} ~~\text{if $\alpha$ does not occur in $T$}
\end{array}
\]
\end{definition}

\subsection{Proofs of Medium Characterization}
\label{app:proofs}

The proof of Theorem~\ref{l:ltypesmedp} relies on the following auxiliary proposition:
\begin{restatable}{proposition}{propfuse}
\label{p:composemed}
Let 
\begin{enumerate}
\item  $\Omega; \G; \D_1 \vdash \etherp{G_1}{\tilde{y}}{\tilde{\omega}} :: z:\one[\omega_m]$, 
with $\dom{\D_1} = \{y_{\pt{p}}, y_{\pt{q}_1}, \ldots, y_{\pt{q}_n}\}$
\item $\Omega; \G; \D_2 \vdash y_\pt{p}(m_\pt{p}@\omega_\pt{p}). y_{\pt{q}_1}(m_{\pt{q}_1}@\omega_{\pt{q}_1}).
\cdots. y_{\pt{q}_n}(m_{\pt{q}_n}@\omega_{\pt{q}_n}).\etherp{G_2}{\tilde{m}}{\tilde{\omega}} :: z:\one[\omega_m]$
\end{enumerate}
be two compositional typings.
Then

\vspace{-0.5cm}
{\small
\[
\Omega; \G; \D_1 \ltfuse \D_2 \vdash  \etherp{G_1}{\tilde{y}}{\tilde{\omega}} 
\prfuse 
y_\pt{p}(m_\pt{p}@\omega_\pt{p}). y_{\pt{q}_1}(m_{\pt{q}_1}@\omega_{\pt{q}_1}).
\cdots. y_{\pt{q}_n}(m_{\pt{q}_n}@\omega_{\pt{q}_n}).\etherp{G_2}{\tilde{m}}{\tilde{\omega}}
:: z:\one[\omega_m]
\]}
is a compositional typing, where 
the typing environment
$\Delta_1 \ltfuse \Delta_2$ is defined as follows:

\vspace{-0.3cm}
{\small\[
\Delta_1 \ltfuse \Delta_2(c) = 
\begin{cases} 
c:\lt{T}[\omega] & \text{if $c \in \dom{\Delta_j}$ and $c \not\in \dom{\Delta_i}$, with $i,j \in \{1,2\}, i \neq j$}
\\
c:\lt{T_1 \ltfuse T_2}[\omega] & \text{if $c:\lt{T_1}[\omega] \in \Delta_1$ and $c:\lt{T_2}[\omega] \in \Delta_2$}
\end{cases}
\]}
\end{restatable}

\begin{proof}[Proof (Sketch)]
We must prove the existence of a typing derivation for the resulting fused process. 
We start by observing that the compositional typing  for $\etherp{G_1}{\tilde{y}}{\tilde{\omega}}$ ensures that its associated derivation will contain one or more occurrences 
of the sequent
\begin{equation}
\Omega; \G;  
y_{\pt{p}}:\one[\omega_\pt{p}], y_{\pt{q}_1}:\one[\omega_{\pt{q}_1}], \ldots, y_{\pt{q}_n}:\one[\omega_{\pt{q}_n}]
\vdash \zero :: z:\one[\omega_m]
\label{eq:fuse1}
\end{equation}
corresponding to one or more occurrences of $\gend$ in $G_1$ (possible because of labeled choices in ${G_1}$): indeed, by \defref{d:ether} we have
$\etherp{\gend}{\tilde{y}}{\tilde{\omega}} = \zero$
and 
by 
\defref{d:loclogt}
we have $\lt{\lend} = \one$.
Observe that the compositional typing 
for 
\begin{equation}
y_\pt{p}(m_\pt{p}@\omega_\pt{p}). y_{\pt{q}_1}(m_{\pt{q}_1}@\omega_{\pt{q}_1}).
\cdots. y_{\pt{q}_n}(m_{\pt{q}_n}@\omega_{\pt{q}_n}).\etherp{G_2}{\tilde{m}}{\tilde{\omega}}
\label{eq:fuse2}
\end{equation}
ensures that 
$\{y_{\pt{p}}, y_{\pt{q}_1}, \ldots, y_{\pt{q}_n}\} \subseteq \dom{\D_2}$, i.e., 
$\D_2$ contains judgements for at least the names in $\dom{\D_1}$---it may also contain other judgments, corresponding to participants that intervene in $G_2$ but not in the sub-protocol $G_1$. 
Given this, the compositional typing for the fused process
\[
\etherp{G_1}{\tilde{y}}{\tilde{\omega}} 
\prfuse 
y_\pt{p}(m_\pt{p}@\omega_\pt{p}). y_{\pt{q}_1}(m_{\pt{q}_1}@\omega_{\pt{q}_1}).
\cdots. y_{\pt{q}_n}(m_{\pt{q}_n}@\omega_{\pt{q}_n}).\etherp{G_2}{\tilde{m}}{\tilde{\omega}}
\]
is obtained by ``stacking up'' the typing derivation for \eqref{eq:fuse2} exactly on the occurrences of sequents of the form \eqref{eq:fuse1} in the typing derivation for  $\etherp{G_1}{\tilde{y}}{\tilde{\omega}}$. 
This is fully consistent with definitions of fusion for processes 
(\defref{d:prfuse})
and local types (\defref{d:ltfuse}): 
the former decrees that $\zero \prfuse P = P$ whereas the latter decrees that $\lend \ltfuse T = T$.
In the resulting ``stacked'' typing derivation, the types for $y_{\pt{p}}, y_{\pt{q}_1}, \ldots, y_{\pt{q}_n}$ that correspond to the behavior of $\etherp{G_1}{\tilde{y}}{\tilde{\omega}} $
can be derived exactly as in the derivation of the first assumption, now starting from the types  $\D_2(y_{\pt{p}}), \D_2(y_{\pt{q}_1}), \ldots, \D_2(y_{\pt{q}_n})$ 
rather than from $\one$. 
\end{proof}


\thmtypesmedp*

\begin{proof} 
By induction on the structure of $G$. 
There are three cases. 

\begin{itemize}
\item The base case, $G = \gend$, is immediate as there are no participants. 

\item The case $G = \gto{p_1}{p_2}\{\lb{l}_i\langle U_i\rangle.G^i\}_{i \in I}$ is exactly as in~\cite{DBLP:conf/forte/CairesP16}, but we report it here for the sake of completeness.  
By the well-formedness assumption (\defref{d:wfltypes}), local types $\proj{G}{\pt{p}_1}, \ldots, \proj{G}{\pt{p}_n}$
are all defined. 
Writing $\mathtt{p}$ and $\mathtt{q}$ instead of
$ \mathtt{p}_1$ and $ \mathtt{p}_2$,
by \defref{d:proj} we have:
\begin{eqnarray}
\proj{G}{\pt{p}} & = & \pt{p}!\{\lb{l}_i\langle U_i\rangle.\proj{G^i}{\pt{p}}\}_{i \in I} \label{neq:p00000} \\
\proj{G}{\pt{q}} & = & \pt{p}?\{\lb{l}_i\langle U_i\rangle.\proj{G^i}{\pt{q}}\}_{i \in I} \label{neq:p0000}\\
\proj{G}{\pt{p}_j} & = & \fuse_{i \in I} \, \proj{G^i}{\pt{p}_j} \quad\text{for  every $j \in  \{3,\ldots, n\}$}  \label{eq:p000}
\end{eqnarray}
We need to show that, for some $\Omega$ and $\G$, 
\begin{equation}
\Omega; \G; c_\pt{p}{:}\lt{\proj{G}{\pt{p}}}[\omega_\pt{p}], 
c_\pt{q}{:}\lt{\proj{G}{\pt{q}}}[\omega_\pt{q}], 
\D
\vdash \etherp{G}{\tilde{c}}{\tilde{\omega}}  ::z:\one[\omega_m]
\label{eq:p0}
\end{equation}
is a compositional typing, with 
$D = c_{\pt{p}_3}{:}\lt{\proj{G}{\pt{p}_3}}[\omega_{\pt{p}_3}], \cdots, 
c_{\pt{p}_n}{:}\lt{\proj{G}{\pt{p}_n}}[\omega_{\pt{p}_n}]$. 

Without loss of generality, we detail the case  $I = \{1,2\}$. 
By \defref{d:ether}, we have:
$$
\etherp{G}{\tilde{c}}{\tilde{\omega}} = c_\mathtt{p}\triangleright\! \begin{cases}
\lb{l}_1 : c_\mathtt{p}(u).\mysel{c_\mathtt{q}}{\lb{l}_1};\outp{c_\mathtt{q}}{v}.( \linkr{u}{v} \para \etherp{G^1}{\tilde{c}}{\tilde{\omega}}) \label{eq:p2}\\
\lb{l}_2 : c_\mathtt{p}(u).\mysel{c_\mathtt{q}}{\lb{l}_2};\outp{c_\mathtt{q}}{v}.( \linkr{u}{v} \para \etherp{G^2}{\tilde{c}}{\tilde{\omega}}) 
\end{cases}
$$
and 
by combining \eqref{neq:p00000} and \eqref{neq:p0000} with \defref{d:loclogt}  we have:
\begin{align*}
\lt{\proj{G}{\pt{p}}}  &= \myoplus{\lb{l}_1: \lt{U_1} \otimes \lt{\proj{G^1}{\pt{p}}} \,,\,\lb{l}_2:\lt{U_2} \otimes \lt{\proj{G^2}{\pt{p}}} }{i \in I} 
\\
\lt{\proj{G}{\pt{q}}}  &=
 \mywith{\lb{l}_1: \lt{U_1} \lolli \, \lt{\proj{G^1}{\pt{q}}}\,,\,\lb{l}_2: \lt{U_2} \lolli \, \lt{\proj{G^2}{\pt{q}}}}{i \in I} 
\end{align*}
Now, 
by assumption $G$ is WF; then, by  construction, 
both $G^1$ and $G^2$ are WF too.
Therefore, by using IH twice we may infer that both
\begin{align}
\Omega; \G; c_\pt{p}{:}\lt{\proj{G^1}{\pt{p}}}[\omega_\pt{p}],\, 
c_\pt{q}{:}\lt{\proj{G^1}{\pt{q}}}[\omega_\pt{q}],\, 
\D_1
& \vdash \!\etherp{G^1}{\tilde{c}}{\tilde{\omega}}::z:\one[\omega_m]  \label{neq:p2}
\\
\Omega; \G; c_\pt{p}{:}\lt{\proj{G^2}{\pt{p}}}[\omega_\pt{p}],\, 
c_\pt{q}{:}\lt{\proj{G^2}{\pt{q}}}[\omega_\pt{q}],\, 
\D_2
& \vdash \!\etherp{G^2}{\tilde{c}}{\tilde{\omega}}::z:\one[\omega_m]  \label{neq:p3}
\end{align}
are compositional typings, for any $\Omega$ and $\G$, with 
\begin{align*}
\D_1 & = c_{\pt{p}_3}{:}\lt{\proj{G^1}{\pt{p}_3}}[\omega_{\pt{p}_3}], \ldots, 
c_{\pt{p}_n}{:}\lt{\proj{G^1}{\pt{p}_n}}[\omega_{\pt{p}_n}]
\\
\D_2 & = c_{\pt{p}_3}{:}\lt{\proj{G^2}{\pt{p}_3}}[\omega_{\pt{p}_3}], \ldots, 
c_{\pt{p}_n}{:}\lt{\proj{G^2}{\pt{p}_n}}[\omega_{\pt{p}_n}]
\end{align*}

Now, to obtain a compositional typing for $\etherp{G}{\tilde{c}}{\tilde{\omega}}$, we must  consider that
$\D_1$ and $\D_2$ may not be identical.
This  is due to the merge-based well-formedness assumption, which admits non identical behaviors in branches $G^1$ and $G^2$
in the case of (local) branching types (cf. \defref{d:mymerg}).

We proceed by induction on $k$, defined as the size of $\D_1$ and $\D_2$ (note that $k = n -2$).
\begin{enumerate} 
\item (Case $k = 0$): Then $\D_1 = \D_2 = \emptyset$ and $\pt{p}$ and $\pt{q}$ are the only participants  in $G$.
Let us write $A_\pt{q}$ to stand for the session type 
$$\mywith{\lb{l}_1 :\lt{U_1} \lolli \lt{\proj{G^1}{\pt{q}}} \,,\, \lb{l}_2 :\lt{U_2} \lolli \lt{\proj{G^2}{\pt{q}}}}{}$$

Based on \eqref{neq:p2} and \eqref{neq:p3}, following the derivation in \figref{fig:ltypesmed},
we may derive typings 
{\small
\begin{align}
\Omega; \G; c_\mathtt{p} : \lt{U_1} \otimes \lt{\proj{G^1}{\pt{p}}}[\omega_\pt{p}], \, c_\mathtt{q}{:}A_\pt{q}[\omega_\pt{q}] \vdash 
      c_\mathtt{p}(u).\mysel{c_\mathtt{q}}{\lb{l}_1};\outp{c_\mathtt{q}}{v}.(\linkr{u}{v} \para  \etherp{G^1}{\tilde{c}}{\tilde{\omega}})  ::z:\one[\omega_m]
      \label{neq:p4}\\
      \Omega; \G; c_\mathtt{p} : \lt{U_2} \otimes \lt{\proj{G^2}{\pt{p}}}[\omega_\pt{p}], \, c_\mathtt{q}{:}A_\pt{q}[\omega_\pt{q}] \vdash 
      c_\mathtt{p}(u).\mysel{c_\mathtt{q}}{\lb{l}_2};\outp{c_\mathtt{q}}{v}.(\linkr{u}{v} \para  \etherp{G^2}{\tilde{c}}{\tilde{\omega}}) ::z:\one[\omega_m]
      \label{neq:p5}
\end{align}
}
The proof for this case is completed using \eqref{neq:p4} and \eqref{neq:p5} as premises for Rule~\name{$\lft\oplus$}:
\[
\Omega; 
 \G; c_\mathtt{p} : \myoplus{\lb{l}_1 : \lt{U_1} \otimes \lt{\proj{G^1}{\pt{p}}} \,,\, \lb{l}_2 : \lt{U_2} \otimes \lt{\proj{G^2}{\pt{p}}}}{}[\omega_\pt{p}], \, c_\mathtt{q}{:}A_\pt{q}[\omega_\pt{p}] \vdash 
      \etherp{G}{\tilde{c}}{\tilde{\omega}} ::z:\one[\omega_m]
\]

\begin{figure}[t!]
{\scriptsize
\[
{
\infer[\name{${\lft\otimes}$}]
{ \Omega; \G; c_\mathtt{p} : \lt{U_1} \otimes \lt{\proj{G^1}{\pt{p}}}[\omega_\pt{p}], \, c_\mathtt{q}{:}\mywith{\lb{l}_i :(\lt{U_i} \lolli \lt{\proj{G^i}{\pt{q}}})}{i \in I}[\omega_\pt{q}]  , \D_1 \vdash 
      c_\mathtt{p}(u).\mysel{c_\mathtt{q}}{\lb{l}_1};\outp{c_\mathtt{q}}{v}.(\linkr{u}{v} \para \etherp{G^1}{\tilde{c}}{\tilde{\omega}} ) ::z:\one[\omega_m]}{
      \infer=[\name{$\lft\with_2$}]{\Omega; \G; u: \lt{U_1}[\omega_\pt{p}],\,  c_\mathtt{p} :  \lt{\proj{G^1}{\pt{p}}}[\omega_\pt{p}],\, c_\mathtt{q}{:}\mywith{\lb{l}_i :(\lt{U_i} \lolli \lt{\proj{G^i}{\pt{q}}})}{i \in I}[\omega_\pt{q}] , \D_1 \vdash \mysel{c_\mathtt{q}}{\lb{l}_1};\outp{c_\mathtt{q}}{v}.( \linkr{u}{v} \para  \etherp{G^1}{\tilde{c}}{\tilde{\omega}} ) \mathstrut ::z:\one[\omega_m]}{
\infer[\name{$\lft\with_1$}]
      {  \Omega; \G; u: \lt{U_1}[\omega_\pt{p}],\,  c_\mathtt{p} :  \lt{\proj{G^1}{\pt{p}}}[\omega_\pt{p}],\, c_\mathtt{q}{:}\mywith{\lb{l}_1 :\lt{U_1} \lolli \lt{\proj{G^1}{\pt{q}}}}{\{1\}}[\omega_\pt{q}] , \D_1 \vdash \mysel{c_\mathtt{q}}{\lb{l}_1};\outp{c_\mathtt{q}}{v}.( \linkr{u}{v} \para  \etherp{G^1}{\tilde{c}}{\tilde{\omega}}  ) ::z:\one[\omega_m]}
      {\infer[\name{${\lft\lolli}$}]
             { \Omega; \G; u: \lt{U_1}[\omega_\pt{p}],\,  c_\mathtt{p} :  \lt{\proj{G^1}{\pt{p}}}[\omega_\pt{p}],\, c_\mathtt{q}{:}\lt{U_1} \lolli \lt{\proj{G^1}{\pt{q}}}[\omega_\pt{q}] , \D_1 \vdash \outp{c_\mathtt{q}}{v}.( \linkr{u}{v} \para  \etherp{G^1}{\tilde{c}}{\tilde{\omega}}  ) ::z:\one[\omega_m]}
             {\infer[\name{$\mathsf{id}$}]{\Omega; \G; u{:}\lt{U_1}[\omega_\pt{p}] \vdash \linkr{u}{v} :: v : \lt{U_1}[\omega_\pt{p}]}{} &
              \infer[\name{${\lft\one}$}]{ \Omega; \G; c_\pt{p}{:}\lt{\proj{G^1}{\pt{p}}}[\omega_\pt{p}],\, c_\pt{q}{:}\lt{\proj{G^1}{\pt{q}}}[\omega_\pt{q}],\,  \D_1 \vdash  \etherp{G^1}{\tilde{c}}{\tilde{\omega}} ::z:\one[\omega_m]}{
            }}}}}
              }
\]
}
\caption{Derivation for $c_\mathtt{p}(u).\mysel{c_\mathtt{q}}{\lb{l}_1};\outp{c_\mathtt{q}}{v}.(\linkr{u}{v} \para  \etherp{G^1}{\tilde{c}}{\tilde{\omega}} )$.\label{fig:ltypesmed}}
\end{figure}

\item (Case $k > 0$): Then there exists
a participant  $\pt{p}_k$, 
 types 
 $B_1 = \lt{\proj{G^1}{\pt{p}_k}}$, 
 $B_2  = \lt{\proj{G^2}{\pt{p}_k}}$ and environments
$\D'_1, \D'_2$ such that 
$\D_1 = c_{\pt{p}_k}{:}B_1[\omega_{\pt{p}_k}], \D'_1$ 
and $\D_2 = c_{\pt{p}_k}{:}B_2[\omega_{\pt{p}_k}], \D'_2$.

By induction hypothesis, 
there is a compositional typing starting from 
\begin{align*}
\Omega; \G; c_\pt{p}{:}\lt{\proj{G^1}{\pt{p}}}[\omega_\pt{p}],\, 
c_\pt{q}{:}\lt{\proj{G^1}{\pt{q}}}[\omega_\pt{q}],\, {\D'_1} & \vdash \!\etherp{G^1}{\tilde{c}}{\tilde{\omega}}  ::z:\one[\omega_m] \\
\Omega; \G; c_\pt{p}{:}\lt{\proj{G^2}{\pt{p}}}[\omega_\pt{p}],\, 
c_\pt{q}{:}\lt{\proj{G^2}{\pt{q}}}[\omega_\pt{q}],\, {\D'_2} & \vdash \!\etherp{G^2}{\tilde{c}}{\tilde{\omega}}  ::z:\one[\omega_m] 
\end{align*}
 resulting into 
$$
 \Omega; \G; c_\mathtt{p} : \myoplus{\lb{l}_1 : \lt{U_1} \otimes \lt{\proj{G^1}{\pt{p}}} \,,\, \lb{l}_2 : \lt{U_2} \otimes \lt{\proj{G^2}{\pt{p}}}}{}[\omega_\pt{p}], \, c_\mathtt{q}{:}A_\pt{q}[\omega_\pt{q}], \D'_1 \vdash 
      \etherp{G}{\tilde{c}}{\tilde{\omega}}  ::z:\one[\omega_m]
$$
since $\D'_1 = \D'_2$.

To extend the typing derivation to $\D_1$ and $\D_2$, 
we proceed by a case analysis on the shape of $B_1$ and $B_2$.
 We aim to show that either (a) $B_1$ and $B_2$ are already identical base or session  types or (b) that 
typing allows us to transform them into identical types. 
We rely  on the definition of $\fuse$ (\defref{d:mymerg}). 
There are two main sub-cases:
\begin{enumerate} 
\item Case $B_1 \neq \mywith{\lb{l}_h: A_h}{h \in H}$: 
Then, since \defref{d:mymerg} decrees $T \fuse T = T$
and the fact that merge-based well-definedness depends on $\fuse$, 
we may infer 
$B_2 = B_1$.
Hence, $\D_1 = \D_2$ and the
desired derivation is obtained as in the base case.


\item Case $B_1 = \mywith{\lb{l}_h: A_h}{h \in H}$: This is the interesting case: 
even if 
merge-based well-formedness of $G$
ensures that both 
$B_1$ and $B_2$
are    selection types, they may not be identical.
If $B_1$ and $B_2$ are identical then we proceed as in previous sub cases.
Otherwise, then 
due to $\fuse$
there are some finite number of labeled alternatives in $B_1$ but not in $B_2$ and/or viceversa.
Also, Def.~\ref{d:mymerg} ensures that common options (if any) are identical in both branches.
We may then use Rule \name{$\lft\with_2$} (cf. Appendix~\ref{app:trules}) to ``complement'' occurrences of types 
$B_1$ and $B_2$ in \eqref{neq:p2} and \eqref{neq:p3}
as appropriate to make them coincide and achieve identical typing. 
Rule \name{$\lft\with_2$} is silent; as labels are finite, this completing task is also finite,  and results into $\D_1 = \D_2$. 
\end{enumerate}
\end{enumerate}

\item
Finally, we have the case 
$G = \gmoves{\pt{p}}{\pt{q}_1, \ldots, \pt{q}_n}{w}{G_1}{G_2}$. 
Without loss of generality, we consider the global type
$G = \gmoves{\pt{p}}{\pt{q}}{w}{G_1}{G_2}$, i.e., the type in which the sub-protocol $G_1$ only involves two participants, namely $\pt{p}$ and  $\pt{q}$.
By the well-formedness assumption, local types $\proj{G}{\pt{p}}$, $\proj{G}{\pt{q}}$, $\proj{G}{\pt{p}_3}, \ldots, \proj{G}{\pt{p}_n}$ are all defined. 
In particular, by \defref{d:proj} we have:
\begin{eqnarray}
\proj{G}{\pt{p}} & = & \here{\beta.(\exists \alpha. @_\alpha\, \proj{G_1}{\pt{p}}) \ltfuse @_\beta\, \proj{G_2}{\pt{p}}} \label{eq:p00000} \\
\proj{G}{\pt{q}} & = & \here{\beta.(\forall \alpha. @_\alpha\, \proj{G_1}{\pt{q}}) \ltfuse @_\beta\, \proj{G_2}{\pt{q}}} \label{eq:p0000}
\end{eqnarray}
We need to show that, for some $\G$, 
\begin{equation}
\Omega; \G; c_\pt{p}{:}\lt{\proj{G}{\pt{p}}}[\omega_\pt{p}],~ 
c_\pt{q}{:}\lt{\proj{G}{\pt{q}}}[\omega_\pt{q}],~ 
\D\subst{\tilde{c}}{\tilde{m}}
\vdash \etherp{G}{\tilde{c}}{\tilde{\omega}}  ::z:\one[\omega_m]
\label{eq:p0}
\end{equation}
with $\D\subst{\tilde{c}}{\tilde{m}} = c_{\pt{p}_3}{:}\lt{\proj{G}{\pt{p}_3}}[\omega_3],~ \ldots, 
~c_{\pt{p}_n}{:}\lt{\proj{G}{\pt{p}_n}}[\omega_n]$ 
is a compositional typing.
By \defref{d:ether}, we have:
\begin{align}
\etherp{G}{\tilde{c}}{\tilde{\omega}} = &
c_\pt{p}(\alpha).c_{\pt{q}}\out{\alpha}.
c_\pt{p}(y_\pt{p}@\alpha).c_{\pt{q}}(y_{\pt{q}}@\alpha). 
\label{eq:medg}
\\
&
\quad \etherp{G_1}{\tilde{y}}{\tilde{\omega}\{\alpha/\omega_{\pt{p}},
   \alpha/\omega_{\pt{q}}\} } \prfuse y_\pt{p}(m_\pt{p}@\omega_\pt{p}). y_{\pt{q}}(m_{\pt{q}}@\omega_{\pt{q}}). \etherp{G_2}{\tilde{m}}{\tilde{\omega}}
   \nonumber
\end{align}
and 
by combining \eqref{eq:p00000} and \eqref{eq:p0000} with \defref{d:loclogt}  we have:
\begin{eqnarray*}
\lt{\proj{G}{\pt{p}}}  &=&  \here{\beta.\exists \alpha. @_\alpha \lt{\proj{G_1}{\pt{p}} \ltfuse @_\beta\, \proj{G_2}{\pt{p}}}}\\
\lt{\proj{G}{\pt{q}}}  &=&  \here{\beta.\forall \alpha. @_\alpha \lt{\proj{G_1}{\pt{q}} \ltfuse @_\beta\, \proj{G_2}{\pt{q}}}}\end{eqnarray*}
Now, 
by assumption $G$ is WF; then, by construction
both $G_1$ and $G_2$ are WF too.
Therefore, by using IH twice we may infer that both
\begin{align}
\Omega; \G; y_\pt{p}{:}\lt{\proj{G_1}{\pt{p}}}[\alpha],\, 
y_\pt{q}{:}\lt{\proj{G_1}{\pt{q}}}[\alpha] 
& \vdash 
\!\etherp{G_1}{\tilde{y}}{\tilde{\omega}\{\alpha/\omega_{\pt{p}},
   \alpha/\omega_{\pt{q}}\} } ::z:\one[\omega_m] \label{eq:p2}
   \\
\Omega; \G; m_\pt{p}{:}\lt{\proj{G_2}{\pt{p}}}[\omega_\pt{p}],\, 
m_\pt{q}{:}\lt{\proj{G_2}{\pt{q}}}[\omega_\pt{q}],\, 
\D 
&\vdash 
\etherp{G_2}{\tilde{m}}{\tilde{\omega}} ::z:\one[\omega_m] \label{eq:p3}
\end{align}
are compositional typings, for any $\Omega$, $\G$, 
with 
$\D$ as above.
Using Rule \name{$\lft @$} twice, assuming 
$\alpha \prec \omega_1$ and $\alpha \prec \omega_2$,
from \eqref{eq:p3} we can derive:
\begin{eqnarray}
\Omega; \G; 
y_\pt{p}{:}@_{\omega_1}\,\lt{\proj{G_2}{\pt{p}}}[\alpha],\, 
y_\pt{q}{:}@_{\omega_2}\,\lt{\proj{G_2}{\pt{q}}}[\alpha],\, 
{\D} \vdash 
y_\pt{p}(m_{\pt{p}}@\omega_{\pt{p}}).y_\pt{q}(m_{\pt{q}}@\omega_{\pt{p}}).\etherp{G_2}{\tilde{m}}{\tilde{\omega}}  \label{eq:p4}
\end{eqnarray}

Using \propref{p:composemed} on \eqref{eq:p2} and \eqref{eq:p4} we obtain:
\begin{eqnarray}
\Omega; \G; 
y_\pt{p}{:}\lt{\proj{G_1}{\pt{p}} \ltfuse @_{\omega_1}\, \proj{G_2}{\pt{p}}}[\alpha],\, 
y_\pt{q}{:}\lt{\proj{G_1}{\pt{q}} \ltfuse @_{\omega_2}\,\proj{G_2}{\pt{q}}}[\alpha],\, 
{\D}  \nonumber \\
\qquad \vdash  \etherp{G_1}{\tilde{y}}{\tilde{\omega}\{\alpha/\omega_{\pt{p}},
   \alpha/\omega_{\pt{q}}\} } \prfuse
y_\pt{p}(m_{\pt{p}}@\omega_{\pt{p}}).y_\pt{q}(m_{\pt{q}}@\omega_{\pt{p}}).\etherp{G_2}{\tilde{m}}{\tilde{\omega}}  \label{eq:p5}
\end{eqnarray}
We then have the following derivation, which completes the proof for this case:
{\small
\[
\infer[\name{$\lft\downarrow$}\times2]
{
\begin{array}{l}
\Omega; \G; 
c_\pt{p}{:}\here{\beta.(\exists \alpha. @_\alpha.\,\lt{\proj{G_1}{\pt{p}} \ltfuse @_{\beta}\, \proj{G_2}{\pt{p}}})}[\omega_1],\, 
c_\pt{q}{:}\here{\beta.(\forall \alpha. @_\alpha.\,\lt{\proj{G_1}{\pt{q}} \ltfuse @_{\beta}\,\proj{G_2}{\pt{q}}})}[\omega_2],\, 
{\D} 
\\
\vdash 
c_\pt{p}(\alpha).c_{\pt{q}}\out{\alpha}.
c_\pt{p}(y_\pt{p}@\alpha).c_{\pt{q}}(y_{\pt{q}}@\alpha). 
\etherp{G_1}{\tilde{y}}{\tilde{\omega}\{\alpha/\omega_{\pt{p}},
   \alpha/\omega_{\pt{q}}\} } \prfuse y_\pt{p}(m_\pt{p}@\omega_\pt{p}). y_{\pt{q}}(m_{\pt{q}}@\omega_{\pt{q}}). \etherp{G_2}{\tilde{m}}{\tilde{\omega}}
\end{array}
}
{
\infer[\name{$\lft \exists$}]
{
\begin{array}{l}
\Omega; \G; 
c_\pt{p}{:}\exists \alpha. @_\alpha.\,\lt{\proj{G_1}{\pt{p}} \ltfuse @_{\omega_1}\, \proj{G_2}{\pt{p}}}[\omega_1],\, 
c_\pt{q}{:}\forall \alpha. @_\alpha.\,\lt{\proj{G_1}{\pt{q}} \ltfuse @_{\omega_2}\,\proj{G_2}{\pt{q}}}[\omega_2],\, 
{\D} 
\\
\vdash 
c_\pt{p}(\alpha).c_{\pt{q}}\out{\alpha}.
c_\pt{p}(y_\pt{p}@\alpha).c_{\pt{q}}(y_{\pt{q}}@\alpha). 
\etherp{G_1}{\tilde{y}}{\tilde{\omega}\{\alpha/\omega_{\pt{p}},
   \alpha/\omega_{\pt{q}}\} } \prfuse y_\pt{p}(m_\pt{p}@\omega_\pt{p}). y_{\pt{q}}(m_{\pt{q}}@\omega_{\pt{q}}). \etherp{G_2}{\tilde{m}}{\tilde{\omega}}
\end{array}
}{
\infer[\name{$\lft \forall$}]{
\begin{array}{l}
\Omega, \omega_1 \prec \alpha; \G; 
c_\pt{p}{:}@_\alpha.\,\lt{\proj{G_1}{\pt{p}} \ltfuse @_{\omega_1}\, \proj{G_2}{\pt{p}}}[\omega_1],\, 
c_\pt{q}{:}\forall \alpha. @_\alpha.\,\lt{\proj{G_1}{\pt{q}} \ltfuse @_{\omega_2}\,\proj{G_2}{\pt{q}}}[\omega_2],\, 
{\D} 
\\
\vdash 
c_{\pt{q}}\out{\alpha}.
c_\pt{p}(y_\pt{p}@\alpha).c_{\pt{q}}(y_{\pt{q}}@\alpha). 
\etherp{G_1}{\tilde{y}}{\tilde{\omega}\{\alpha/\omega_{\pt{p}},
   \alpha/\omega_{\pt{q}}\} } \prfuse y_\pt{p}(m_\pt{p}@\omega_\pt{p}). y_{\pt{q}}(m_{\pt{q}}@\omega_{\pt{q}}). \etherp{G_2}{\tilde{m}}{\tilde{\omega}}
\end{array}
}{
\infer[\name{$\lft @$}]
{
\begin{array}{l}
\Omega, \omega_1 \prec \alpha, \omega_2 \prec \alpha; \G; 
c_\pt{p}{:}@_\alpha.\,\lt{\proj{G_1}{\pt{p}} \ltfuse @_{\omega_1}\, \proj{G_2}{\pt{p}}}[\omega_1],\, 
y_\pt{q}{:}@_\alpha.\,\lt{\proj{G_1}{\pt{q}} \ltfuse @_{\omega_2}\,\proj{G_2}{\pt{q}}}[\omega_2],\, 
{\D} 
\\
\vdash 
c_\pt{p}(y_\pt{p}@\alpha).c_{\pt{q}}(y_{\pt{q}}@\alpha). 
\etherp{G_1}{\tilde{y}}{\tilde{\omega}\{\alpha/\omega_{\pt{p}},
   \alpha/\omega_{\pt{q}}\} } \prfuse y_\pt{p}(m_\pt{p}@\omega_\pt{p}). y_{\pt{q}}(m_{\pt{q}}@\omega_{\pt{q}}). \etherp{G_2}{\tilde{m}}{\tilde{\omega}}
\end{array}
}{
\infer[\name{$\lft @$}]
{
\begin{array}{l}
\Omega, \omega_1 \prec \alpha,  \omega_2 \prec \alpha; \G; 
y_\pt{p}{:}\lt{\proj{G_1}{\pt{p}} \ltfuse @_{\omega_1}\, \proj{G_2}{\pt{p}}}[\alpha],\, 
c_\pt{q}{:}@_\alpha.\,\lt{\proj{G_1}{\pt{q}} \ltfuse @_{\omega_2}\,\proj{G_2}{\pt{q}}}[\omega_2],\, 
{\D} 
\\
\vdash 
c_{\pt{q}}(y_{\pt{q}}@\alpha). 
\etherp{G_1}{\tilde{y}}{\tilde{\omega}\{\alpha/\omega_{\pt{p}},
   \alpha/\omega_{\pt{q}}\} } \prfuse y_\pt{p}(m_\pt{p}@\omega_\pt{p}). y_{\pt{q}}(m_{\pt{q}}@\omega_{\pt{q}}). \etherp{G_2}{\tilde{m}}{\tilde{\omega}}
\end{array}
}{
\eqref{eq:p5}
}
}
}
}
}
\]
}
\end{itemize}
\end{proof}


The proof of the converse of \thmref{l:ltypesmedp} proceeds similarly; we must take into account that, given a global type $G$, the process structure of 
$\etherp{G}{\tilde{c}}{\tilde{\omega}}$
will induce types closely related to $\proj{G}{{\pt{p_1}}}, \ldots, \proj{G}{{\pt{p_1}}}$, up to occurrences of two type operators whose typing rules enforce a silent interpretation of processes, namely \name{$\lft\with_2$} and \name{$\lft\downarrow$}.

\thmmedltypes*
\begin{proof}[Proof (Sketch)]
By induction on the structure of $G$, following the lines of the proof of \thmref{l:ltypesmedp}, exploiting  $\subt$ (cf. \defref{def:pctypes}). There are three cases:
\begin{itemize} 
\item (Case $G = \gend$): Then $\etherp{G}{\tilde{c}}{\tilde{\omega}} = \zero$, $\partp{G} = \emptyset$, and the thesis follows vacuously.
Notice that from the assumption $\Omega; \G; \cdot \vdash \ether{G} $ and Rule \name{$\lft\one$}
 we may derive 
\[
\Omega; \G; c_{\pt{r}_j}{:}\one \vdash \ether{G}::z:\one[\omega_m]
\]
for any $\pt{r}_j$.
In that case, observe that \defref{d:proj} decrees that 
$\proj{\gend}{\mathtt{r}_j} = \lend$, for any $\mathtt{r}_j$.
There could be spurious occurrences of ``here'', as in, e.g,, $\here{\alpha_1. \cdots \here{\alpha_n.\one}}$.
The thesis holds using $\subt$, for $\lt{\lend} = \one$ (cf. \defref{d:loclogt}).

\item (Case $G = \gto{p_1}{p_2}\{\lb{l}_i\langle U_i\rangle.G^i\}_{i \in I}$) This case follows by typing inversion on the structure of $\etherp{G}{\tilde{c}}{\tilde{\omega}}$, using the derivation presented in the second case of the proof of \thmref{l:ltypesmedp}. The main aspect to consider are the possible uses of the silent Rule~\name{$\lft\with_2$}: to prove the correspondence between $\lt{\proj{G}{\pt{p_2}}}$ and the branching type $A_2$, we exploit the first axiom of $\subt$ (cf. \defref{def:pctypes}) to appropriately handle/prune silently added alternatives in the branching. We use $\subt$ to relate $A_2$ and $\lt{\proj{G}{\pt{p_2}}}$ up to occurrences of ``here'' as above.

\item (Case $G = \gmoves{\pt{p}}{\pt{q}_1, \ldots, \pt{q}_n}{w}{G_1}{G_2}$) This case also follows by typing inversion on the structure of $\etherp{G}{\tilde{c}}{\tilde{\omega}}$, using the derivation presented in the third case of the proof of \thmref{l:ltypesmedp}. The main aspect to consider for the intended correspondence is that Rule~\name{$\lft\downarrow$} silently induces type constructs of the form $\here{\alpha.A}$ within the types for $\etherp{G}{\tilde{c}}{\tilde{\omega}}$. To obtain the correspondence, we exploit the second and third axioms of $\subt$, which make explicit required ``here'' operators in the type and remove spurious ``heres'', respectively.
\end{itemize}
\end{proof}

\section{Extended Examples}\label{sec:examples}

\subsection{Negotiation Procedure}

This example is adapted from \cite{DBLP:conf/concur/DemangeonH12}, consisting
of a negotiation procedure {\it Nego} between two participants of a
three-party interaction.
%
The negotiation consists of an agreement
on a contract: one participant specifies a request, while the other
offers a corresponding contract. The first participant may either
accept the contract and end the protocol or make a counter-offer. For
the sake of conciseness we assume that the counter-offer is accepted:

\begin{tabbing}
  $\mathit{Nego}_{\m{p},\m{q}} \triangleq$ \=
  $\gto{p}{q}\{\lb{ask}\langle \mathit{terms}\rangle.$
  $\gto{q}{p}\{
  \lb{proposition}\langle\mathit{contract}_1\rangle.$\\
   \>$\gto{p}{q}\{$\=$\lb{accept}.\gend ,$
   $\lb{counter}\langle\mathit{contract}_1\rangle.$
   $\gto{q}{p}\{\lb{accept}.\gend\}\} \}\}$

\end{tabbing}

The main protocol consists of a client, an agent
and an instrument, each initially in their own domains.
The client first sends a request to the
agent for some instrument they wish to use. The agent connects to
the instrument which acknowledges when available. The agent then
enters the negotiation sub-protocol with the client (via protocol
$\mathit{Nego}$), by having both agent and the client migrate
to domain $d_n$. This movement models the
 trusted setting at which the agent and the
client coexist in order to successfully negotiate the instrument usage.
After the negotiation stage is complete, both the client and the
instrument migrate to a common domain $d_i$ to perform the rest of the
protocol, which for the sake of conciseness we model with the client
either aborting the interaction or sending a command to the instrument
and then receiving back the appropriate result:
\begin{tabbing}
  $\gto{\m{client}}{\m{agent}}
  \{$ \= $\lb{req}\langle\mathit{coord}\rangle.$
  $\gto{\m{agent}}{\m{instr}}\{ \lb{connect}.$
  $\gto{\m{instr}}{\m{agent}}\{\lb{available}.$
  $\gto{\m{agent}}{\m{client}}\{\lb{ack}.$\\
  \>$\gmoves{\m{agent}}{\m{client}}{d_n}{\mathit{Nego}_{\m{agent},\m{client}}}
  {}$\\
  \>$\gmoves{\m{client}}{\m{instr}}{d_i}{\gto{\m{client}}{\m{instr}}\{$ $\lb{abort}.\gend
  ,$\\
  \>$\lb{command}\langle\mathit{code}\rangle .
        \gto{\m{instr}}{\m{client}}\{\lb{result}\langle\mathit{data}\rangle.\gend$
  $\}\}\}\}\}\}}{\gend}$
\end{tabbing}

By leveraging our notion of medium, we can make explicit
the fact that the three participants are distributed agents, each
located at independent domains that can access the medium substrate
(i.e.~the domain of the medium).
Through the medium-orchestrated
interaction, the use of domain migration primitives enables us to
explicitly model the various domain movement steps that
the participants must follow to implement the 
protocol. This is in sharp contrast with more traditional approaches
to multiparty protocols \cite{DBLP:conf/popl/HondaYC08}, where such
domain specific notions are implicit.
The medium for the global type above is (we assume that
the three participants initially reside at worlds 
$w_\pt{client}, w_\pt{agent}, w_\pt{instr}$, respectively):

\vspace{-0.2cm}
{\small\begin{tabbing}
  $\mycasebig{c_\pt{client}}{$\=$\lb{req} :
    c_\pt{client}(u).\mysel{c_\pt{agent}}{\lb{req}};
    \outp{c_\pt{agent}}{v}.(\linkr{u}{v} \mid$\\
    \>$\mycasebig{c_\pt{agent}}{$\=$\lb{connect} :
      \mysel{c_\pt{instr}}{\lb{connect}};
      \mycasebig{c_\pt{instr}}{$\=$\lb{available} :
        \mysel{c_\pt{agent}}{\lb{available}};$\\
        \>$\mycasebig{c_\pt{agent}}{$\=$\lb{ack} :
          \mysel{c_\pt{client}}{\lb{ack}};
          c_\pt{agent}(d_n). c_\pt{client}\langle d_n \rangle.
          c_\pt{agent}(y_\pt{agent}@d_n). c_\pt{client}(y_\pt{client}@d_n).$\\
        \>$\etherp{\mathit{Nego}}{y_\pt{agent},y_\pt{client}}{\tilde{d_n}}
      \prfuse
      (y_\pt{agent}(z_\pt{agent}@w_\pt{agent}). y_\pt{client}(z_\pt{client}@w_\pt{client}).$\\
      \>\qquad$z_\pt{client}(d_i). c_\pt{instr}\langle d_i\rangle.
      z_\pt{client}(y_\pt{client}@d_i).c_\pt{instr}(y_\pt{instr}@d_i).$\\
      \>\qquad$\mycasebig{y_\pt{client}}{$\=$ \lb{abort} : \mysel{y_\pt{instr}}{\lb{abort}};\mathbf{0} , $\\
        \>\>$\lb{command} :
        y_\pt{client}(c).\mysel{y_\pt{instr}}{\lb{command}};
        \outp{y_\pt{instr}}{d}.(\linkr{c}{d} \mid $\\
        \>\>$\mycasebig{y_\pt{instr}}
        {\lb{result} :
          y_\pt{instr}(data). \mysel{y_\pt{client}}{\lb{result}};
          \outp{y_\pt{client}}{r}.(\linkr{data}{r} \mid $\\
          \>\>$y_\pt{client}(z_\pt{client}@w_\pt{client}).y_\pt{instr}(z_\pt{instr}@w_\pt{instr}
          ).\mathbf{0})}{} )
      }{}   )}{}}{}}{})}{}$
\end{tabbing}}
\vspace{-0.1cm}

The first two lines of the medium definition correspond to the
initial exchange between the client, the agent and the instrument.
The actions
after the emission of label $\lb{ack}$ to the client model the
migration protocol: the agent emits the domain identifier to the
medium, which then forwards it to the client and receives from both
participants the session handles $y_\pt{agent}$ and
$y_\pt{client}$, located at $d_n$. After the
migration takes place, the medium orchestrates the negotation
between the agent and the client using the new session
handles.

After the negotiation, the agent and the client migrate
back to their initial domains $w_\pt{agent}$ and $w_\pt{client}$,
respectively, and the interactions
between the client and the instrument take place: the client and the
instrument migrate to $d_i$, sending to the medium the 
session handles $y_\pt{client}$ and $y_\pt{instr}$, followed by the 
client emitting an $\lb{abort}$ or $\lb{command}$ message which is
forwarded to the instrument. In the latter case, the instrument forwards
the result to the client. Finally, both the client and the
instrument migrate back to their initial domains.

\subsection{Domain-aware Middleware}

A common design pattern in distributed computing is the notion of a
middleware agent which answers requests from clients,
sometimes offloading the requests to some server (e.g.~to better
manage local resource availability). The mediation
between the middleware and the server often involves some form of
privilege escalation or specialized authentication, which we can now
model via domain migration.
We first represent a simple offloading protocol between the middleware
$\pt{p}$ and the server $\pt{q}$:
\begin{tabbing}
  $\mathit{Offload}_{\pt{p},\pt{q}} \triangleq
  \gto{p}{q}\{\lb{req}\langle data
  \rangle. \gto{q}{p}\{\lb{reply}\langle ans \rangle.\gend\}  \}$
 \end{tabbing} 
 The global interaction is represented by the following global type:
 \begin{tabbing}
   $\gto{client}{mw}\{\lb{request}\langle req\rangle.$\\
   $\gto{mw}{client}\{$ \=$\lb{reply}\langle ans \rangle.\gto{mw}{server}\{\lb{done}. \gend\},$\\
   \>$
   \lb{wait}. \gto{mw}{server}\{\lb{init}. \gmoves{\pt{mw}}{\pt{server}
   }{w_\pt{priv}}{\mathit{Offload}_{\pt{mw},\pt{server}} }$\\
   \>${\gto{mw}{client}\{\lb{reply}\langle ans \rangle.\zero\}} \}  \}\}$
   \end{tabbing}
   The medium for this protocol is given by:
   \begin{tabbing}
     $\mycasebig{c_\pt{client}}{\lb{request} : c_\pt{client}(r).
     \mysel{c_\pt{mw}}{\lb{request}};
     \outp{c_\pt{mw}}{v}.(\linkr{r}{v} \mid$\\
     $\mycasebig{c_\pt{mw}}{$ \=$\lb{reply} :
       c_\pt{mw}(a).\mysel{c_\pt{client}}{\lb{reply}};
       \outp{c_\pt{client}}{n}.(\linkr{a}{n} \mid
       \mycasebig{c_\pt{mw}}{\lb{done} : \mysel{c_\pt{server}}{\lb{done}};\mathbf{0} }{} ),$\\
       \>$ \lb{wait} : \mycasebig{c_\pt{mw}}{$\=$\lb{init} :
         \mysel{c_\pt{server}}{\lb{init}};
         c_\pt{mw}(w_\pt{priv}).c_\pt{server}\langle
         w_\pt{priv}\rangle.$  \\
         \>\>$c_\pt{mw}(y_\pt{mw}@ w_\pt{priv}).
         c_\pt{server}(y_\pt{server}@ w_\pt{priv}).
         \etherp{\mathit{Offload}}{y_\pt{mw},y_\pt{server}}{\tilde{w_\pt{priv}}}
         \,\,\prfuse$\\
         \>\>$(y_\pt{mw}(z_\pt{mw}@w_\pt{mw}).
         y_\pt{server}(z_\pt{server}@w_\pt{server}).$\\
         \>\>$\mycasebig{z_\pt{mw}}{\lb{reply} : z_\pt{mw}(a).
         \mysel{c_\pt{client}}{\lb{reply}};\outp{c_\pt{client}}{n}.(\linkr{a}{n}
       \mid \mathbf{0})}{} )
       }{}  }{}  ) }{}$
 \end{tabbing}
 Notice how the client's domain remains fixed throughout the entire
 interaction, regardless of whether or not the middleware chooses to
 interact with the server to fulfil the client request.
 
\subsection{A Secure Communication Domain}

Our previous examples explore the use of domains in a general distributed
setting. Another interesting aspect of our domain-aware typing
discipline is that communication can only take place between accessible
domains. In scenarios where participants are each situated in distinct
domains, domain movement also governs the ability of participants to
interact. For instance, consider the following protocol excerpt:
\begin{tabbing}
  $\gto{client}{store}\{$\=$\lb{purchase}:
  \gmoves{\pt{store}}{\pt{bank}}{\pt{sec}}{\mathit{SecurePay}}$\\
  \>${\gto{store}{client}\{\lb{success}\langle receipt\rangle.\gend,
     \lb{fail}.\gend\}}  \}$
\end{tabbing}
The protocol above is part of the interaction between an online store
and its clients, where after some number of exchanges the client
decides to purchase the contents of their shopping cart. Upon
receiving the $\lb{purchase}$ message, the store is meant to enter a secure
domain $\pt{sec}$ so it can communicate with the bank role to exchange
potentially sensitive data. In a setting where the client, store and
bank exist at different domains and where only the store domain can
access the bank domain, our domain-aware typing discipline ensures that
no direct communication between the bank and the client domain is
possible, with the entirety of the data flows between the client and
the bank (via the medium) being captured by the type and medium
specification.


\section{Examples of Domain-aware Binary Sessions}\label{app:examples}
In this section we present two examples that further illustrate the
novel features of our domain-aware framework for session
communications.

\subsection{E-Commerce Example}\label{ss:ecom}

We revisit the web store example of \S~\ref{sec:hill}.
Recall the refined web store session type:
\begin{align*}
\mathsf{WStore}_{\mathtt{sec}} \triangleq  &\, \mathsf{addCart} \lolli \with \{ \mathtt{buy}: @_{\mathtt{sec}}\,\mathsf{Payment}, \mathtt{quit}: \one \}\\
\mathsf{Payment}   \triangleq & \,
\mathsf{CCNumber} \lolli \oplus\{\mathtt{ok}:(@_{\mathtt{bnk}}\mathsf{Receipt}) \tensor
\one , \mathtt{nok}:\one\}\\
\mathsf{Bank}  \triangleq & \, \mathsf{CCNumber} \lolli
\oplus\{\mathtt{ok} :\mathsf{Receipt} \tensor \one , \mathtt{nok}:\one\}
\end{align*}
The process that implements the bank interface is to be accessible from
the domain of the web store, moving to a secure domain
$\mathtt{bnk}$ before receiving and validating the payment
information.
Thus, the bank process typing can be specified with the following
judgment, recalling that the domain of the web store is $\mathtt{ws}$:
\[
\mathtt{ws} \prec \mathtt{bnk} ; \cdot ; \cdot \vdash B :: @_{\mathtt{bnk}}\,\mathsf{Bank}[\mathtt{ws}]
\]
The web store will then use the bank interface to fulfill its interface:
\[
\cdot ; \cdot ;b{:}@_{\mathtt{bnk}}\,\mathsf{Bank}[\mathtt{ws}] \vdash
\mathsf{Store} :: z{:}\mathsf{WStore}_{\mathtt{sec}}[\mathtt{ws}]
\]

The $@_w$ type constructor allows us to express a very precise form of coordinated
domain migration. For instance, typing ensures that in order
for the store to produce an output of the form
$@_{\mathtt{bnk}}\mathsf{Payment}$ it must first have interacted with the
bank domain: in order to produce an output of $@_{\mathtt{bnk}}$, it
must be the case that $\mathtt{ws}\prec \mathtt{bnk}$, which is only
known to the store process \emph{after} interacting with the bank domain.
Alternatively, consider the
$\mathsf{WStore}_{\mathtt{sec}}$ service interacting with a client process along channel $x$, each in their
own (accessible) domains, $\mathtt{c}$ and $\mathtt{ws}$, respectively. 
Our framework ensures that interactions between the client and the web store enjoy session fidelity, progress, and termination guarantees. 
Concerning domain-awareness, 
by assuming the client chooses to
buy his product selection, we reach a state that is typed
as follows:
\[
\mathtt{c}\prec\mathtt{ws} ; \cdot ;
x{:}@_{\mathtt{sec}}\mathsf{Payment}[\mathtt{ws}] \vdash Client :: z{:}@_{\mathtt{sec}}\one[\mathtt{c}]
\]
At this point, it is \emph{impossible} for a (typed) client to interact with the
behavior that is protected by the trusted domain $\mathtt{sec}$, since
it is not the case that $\mathtt{c}\prec^* \mathtt{sec}$.
This ensures,
e.g., that a client cannot exploit the payment platform of the web
store by accessing the trusted domain in unforeseen ways.
 Formally, no typing derivation of
$\mathtt{c}\prec\mathtt{ws} ; \cdot ; \mathsf{Payment}[\mathtt{sec}]
\vdash Client ::z{:}@_{\mathtt{sec}}\one[\mathtt{c}]$ exists (Theorem~\ref{thrm:NI}).
The client can only communicate in the secure domain \emph{after} the web
store service has migrated accordingly:
\[
\begin{array}{l}
\mathtt{c}\prec\mathtt{ws} , \mathtt{ws}\prec \mathtt{sec} ; \cdot ;
x'{:}\mathsf{Payment}[\mathtt{sec}] \vdash Client' :: z'{:}\one[\mathtt{sec}] \\
\mbox{where} \quad Client \triangleq x(x'@\mathtt{sec}). \overline{z}\langle z'@\mathtt{sec}\rangle. Client'
\end{array}
\]

It is inconvenient (and potentially error-prone) for the
payment domain to be hardwired in the type. We can solve this issue via
existential quantification as shown in the introduction.
\[
\mathsf{WStore}_\exists  \triangleq  \mathsf{addCart} \lolli \with \{ \mathtt{buy}: \exists \alpha.\,@_{\alpha}\,\mathsf{Payment},
\mathtt{quit}: \one\}
\]
As long as accessibility is irreflexive and antisymmetric, the
server-provided payment domain $w$ will not be able to interact with
the initial $\mathtt{public}$ domain of the interaction except as
specified in the $\mathsf{Payment}$ type.

Alternatively, the server can let the client choose a payment domain
by using universal quantification.  Compliant server code will only be
able to communicate in the client-provided payment domain since the
process must be parametric in $\alpha$.
\[
\mathsf{WStore}_\forall \triangleq \mathsf{addCart} \lolli \with \{ \mathtt{buy}: \forall \alpha.\,@_{\alpha}\,\mathsf{Payment},
\mathtt{quit}: \one\}
\]

\subsection{Spatial Distribution as in $\lambda$5}\label{sect:lambda5}

Murphy et al.~\cite{DBLP:conf/lics/VIICHP04} have proposed a
Curry-Howard interpretation of the intuitionistic modal logic
S5~\cite{DBLP:phd/ethos/Simpson94a} to model distributed computation with worlds
as explicit loci for computation.  Accessibility between worlds was
assumed to be reflexive, transitive, and symmetric because each host
on a network should be accessible from any other host.
Murphy~\cite{MurphyPhd} later generalized this to hybrid logic, an
idea also present in \cite{DBLP:conf/esop/JiaW04}, so that propositions
can explicitly refer to worlds.  Computation in this model was
decidedly \emph{sequential}, and a concurrent extension was proposed
as future work.  Moreover, the system presented some difficulties in
the presence of disjunction, requiring a so-called \emph{action at a
  distance} without an explicit communication visible in the
elimination rule for disjunction.

The present system not only generalizes $\lambda 5$ to permit
concurrency through session-typed linearity, but also solves the
problem of action at a distance because all communication is explicit
in the processes.  Due to this issue in the original formulation of
$\lambda 5$, we will not attempt here to give a full, computationally
adequate interpretation of $\lambda 5$ in our system (which would
generalize~\cite{DBLP:conf/fossacs/ToninhoCP12}), but instead explain the spatially
distributed computational interpretation of our type system directly.
\begin{enumerate}[$\bullet$]
\item $c : \Box A$.  Channel $c$ offering $A$ can be used in any
  domain.
\item $c : \Dia A$.  Channel $c$ is offering $A$ in some (hidden)
  domain.
\end{enumerate}
These are mapped into our hybridized linear logic (choosing a fresh $\alpha$ each time) with
\[
\begin{array}{lcl}
  \Box A & = & \forall \alpha. @_\alpha A \\
  \Dia A & = & \exists \alpha. @_\alpha A
\end{array}
\]
A process $P :: c {:} \Box A[w_1]$ will therefore receive, along $c$, a
world $w_2$ accessible from $w_1$ and then move to $w_2$, offering $A$
in domain $w_2$.  Conversely, a process $P :: c {:} \Dia A[w_1]$ will
send a world $w_2$ along $c$ and then move to $w_2$, offering $A$ in
domain $w_2$.  Processes using such channels will behave dually.

We can now understand the computational interpretation of some of the 
distinctive axioms of S5, keeping in mind that  
accessibility should be reflexive, transitive, and symmetric (we make no
distinction between direct accessibility or accessibility requiring
multiple hops).
\begin{enumerate}[$\bullet$]
\item $\mathsf{K}_\Dia :: z{:}\Box (A \lolli B) \lolli \Dia A \lolli \Dia B[w_0]$.  

Here, $A$
is offered along some $c_1$ at some unknown domain $\alpha$ accessible from $w_0$.
Move the offer of $\Box (A \lolli B)[w_0]$ to $\alpha$ as $c_2$ and send
it $c_1$ to obtain $c_2 : B[\alpha]$.  We abstract this as $\Dia B[w_0]$,
which is possible since $w_0 \prec \alpha$. Intuitively, this axiom
captures the fact that a session transformer $A\lolli B$ that can be
used in \emph{any} domain may be combined with a session $A$ offered in
\emph{some} domain to produce a session behavior $B$, itself
in some hidden domain.
\[
\begin{array}{lcl}
\mathsf{K}_\Diamond & \triangleq & z(x).z(y).y(\alpha).
y(c_1@\alpha).x\langle
\alpha \rangle.x(c_2 @\alpha). \\
&& \overline{c_2}\langle v\rangle.( [c_1 \leftrightarrow v] \mid 
z\langle \alpha \rangle.\overline{z}\langle c_3 @ \alpha\rangle.[c_2
\leftrightarrow c_3])
\end{array}
\]

\item $\mathsf{T} :: z{:}\Box A \lolli  A[w_0]$.  

Given a session $x$ that offers $\Box A$ at $w_0$,
we offer $A$ at $w_0$ by appealing to reflexivity, and thus $w_0 \prec
w_0$. This means that we can then receive from $x$ a fresh channel $c_1 @
w_0$, obtaining an ambient session of type $A$ at domain $w_0$ which
we then forward along $z$. The $\mathsf{T}$ axiom captures the fact
that a mobile session can in fact be accessed ``anywhere''.
\[
\begin{array}{lcl}
\mathsf{T} & \triangleq & z(x).x\langle w_0 \rangle.x( c_1 @
w_0).[c_1 \leftrightarrow z]
\end{array}
\]

\item $\mathsf{5} :: z{:}\Diamond A \lolli \Box \Diamond A[w_0]$

Given a session $x$ offering $\Diamond A$ at $w_0$, this means that
$x$ is offering the behavior $A$ at some accessible but unknown
domain $\beta$. We can use this session to provide somewhat of a ``link''
session, that allows any other domain $\alpha$ to also access the behavior $A$ at this
accessible, unknown, domain. We do this by first receiving $\alpha$
along $z$, identifying the domain accessible from $w_0$ which we wish
to link to $\beta$. We then send along $z$ the fresh session $c_1$
located at $\alpha$, along which we shall provide the connection.
We may then receive from $x$ the identity of $\beta$ and a session
$c_2$ located in this domain, send the identity along $c_1$ and a
fresh session $c_3$, located at $\beta$ (only possible due to
the accessibility relation being an equivalence) which we then forward
from $c_2$ as needed.
\[
\begin{array}{lcl}
\mathsf{5} & \triangleq & z(x).z(\alpha).z\langle c_1 @ \alpha
\rangle.x(\beta).x(c_2 @\beta).
c_1\langle \beta \rangle.c_1\langle
c_3 @ \beta \rangle.[c_2 \leftrightarrow c_3]
\end{array}
\]

\end{enumerate}
Other axioms can be given similarly straightforward interpretations
and process realizations.

\end{document}